\documentclass[]{llncs}
\usepackage{xcolor}
\usepackage[T2A]{fontenc}		
\usepackage[utf8]{inputenc}			
\usepackage[english]{babel}	

\usepackage{amsmath,amsfonts,amssymb,mathtools}

\usepackage{wasysym}
\usepackage{amscd}
\usepackage{comment}
\usepackage{listings}
\usepackage{multirow}
\usepackage{amsmath}
\usepackage{graphicx}
\usepackage{wrapfig}
\usepackage{subcaption}
\usepackage{orcidlink}

\newcommand{\tool}{{\textsf{Trans}}\xspace}

\def\AA{\mathcal{A}}

\def\DD{\mathcal{E}}

\def\II{\mathcal{I}}

\def\TT{\mathcal{T}}
\def\JJ{\mathcal{J}}

\def\NN{\mathbb{N}}
\def\lim{{\varprojlim}}

\def\->{\rightarrow}

\usepackage{graphicx}
\usepackage{paralist}
\usepackage{float}
\floatstyle{boxed} 
\restylefloat{figure}
\usepackage{xcolor}
\usepackage{algorithmic}
\usepackage{lipsum}
\usepackage{multicol}
\usepackage{tikz}

\usepackage{listings}

\usepackage{amsmath}

\usepackage{amssymb,amsfonts}

\usepackage{multirow}
\usepackage{boldline}
\usepackage{xspace}
\usepackage{cleveref}

\usepackage{verbatim}  

\begin{document} 

\newcommand{\mc}[1]{{\color{blue}[MC: #1]}}
\newcommand{\og}[1]{{\color{red}[OG: #1]}}

\newcommand{\pred}{p}
\newcommand{\predNL}{p_1,\ldots,p_k}
\newcommand{\predd}{q}
\newcommand{\preddd}{r}
\newcommand{\Preds}{\mathcal{P}}
\newcommand{\Interp}{\mathcal{I}}
\newcommand{\Interpprime}{\mathcal{{I'}}}
\newcommand{\Refut}{\mathcal{R}}
\newcommand{\CHC}{\pi}
\newcommand{\CHCset}{\Pi}
\newcommand{\CHCsubset}{\Delta}
\newcommand{\CHCsubsetprime}{{\Delta'}}
\newcommand{\nextCHCsubset}{\delta}
\newcommand{\CHCtriple}[3]{\langle \{#1\},#2,#3\rangle}
\newcommand{\CHCfact}[2]{\langle \emptyset,#1,#2\rangle}
\newcommand{\CHCrule}{\CHCtriple{\pred}{\varphi}{\predd}}
\newcommand{\CHCruleNL}{\CHCtriple{\predNL}{\varphi}{\predd}}
\newcommand{\CHCquery}{\CHCtriple{\pred}{\varphi}{\bot}}
\newcommand{\CHCqueryNL}{\CHCtriple{\predNL}{\varphi}{\bot}}

\newcommand{\Graph}{G_{\CHCset}}
\newcommand{\vertex}{v_{\pred}}
\newcommand{\vinit}{v_{init}}
\newcommand{\verr}{v_{err}}
\newcommand{\edge}{e_{\CHC}}
\newcommand{\weights}{w}
\newcommand{\CHCvec}{(\CHC_1,\dots,\CHC_n)}
\newcommand{\topo}{\cal{TO}}
\newcommand{\wtopo}{\cal{WTO}}
\newcommand{\ewtopo}{\cal{EWTO}}

\newcommand{\tools}{\textsc{LH}}
\newcommand{\spacer}{{\textsc{Spacer}}\xspace}
\newcommand{\eldarica}{{\textsc{Eldarica}}\xspace}

\newcommand{\Solve}{\mathtt{Solve}}
\newcommand{\InitClauses}{\mathtt{\Select}}
\newcommand{\NextClauses}{\mathtt{\Select}}
\newcommand{\Select}{\mathtt{Select}}
\newcommand{\MkRuleSat}{\Amend}
\newcommand{\Amend}{\mathtt{Amend}}
\newcommand{\AddQ}{\mathtt{Add}}
\newcommand{\PopQ}{\mathtt{Pop}}
\newcommand{\foo}{\mathtt{foo}}
\newcommand{\Shortest}{\mathtt{ShortNt}}

\newcommand{\Queue}{\mathcal{Q}}

\title{$CTL^*$ Verification and Synthesis using\\ 
Existential Horn Clauses} 

\author{Mishel Carelli\,\orcidlink{0009-0000-2181-8205} and Orna Grumberg\,\orcidlink{0009-0005-9682-3312}}
\institute{Technion - Israel Institute of Technology}
\date{}

%

\maketitle

\begin{abstract}
This work proposes a novel approach for automatic verification and synthesis of infinite-state reactive programs with respect to $CTL^*$ specifications, based on translation to Existential Horn Clauses (EHCs). 

$CTL^*$ is a powerful temporal logic, 
which subsumes the temporal logics LTL and CTL, both widely used in specification, verification, and synthesis of complex systems.

EHCs with its solver E-HSF,  
is an extension of Constrained Horn Clauses, which includes
existential quantification as well as the power of handling well-foundedness. 

We develop the translation system $\tool$, which given a verification problem consisting of a program $P$ and a specification $\phi$, builds a set of EHCs which is satisfiable iff  $P$ satisfies $\phi$.
We also develop a synthesis algorithm that given a program with holes in conditions and assignments, fills the holes so that the synthesized program satisfies the given $CTL^*$ specification.

We prove that our verification and synthesis algorithms are both sound and relative complete. Finally, we present case studies to demonstrate the applicability of our algorithms for $CTL^*$ verification and synthesis.

\end{abstract}

\section{Introduction}
This work proposes a novel approach for automatic verification and synthesis of reactive infinite-state programs with respect to $CTL^*$ specifications, based on translation to Existential Horn Clauses (EHC). 

CTL$^*$ is a powerful temporal logic, 
which subsumes the temporal logics LTL and CTL, both widely used in specification, verification, and synthesis of complex systems in industry and in academic research.
In addition to specifying the behavior of a system over time, CTL$^*$ is capable of describing the system's branching structure.

Constrained Horn Clauses (CHCs)  is a fragment of first-order logic which has been recently  highly successful in automated verification of infinite-state programs with respect to safety  properties~\cite{DBLP:conf/sas/BjornerMR13,Bjørner2015}. CHCs include uninterpreted predicates that represent, for instance, loop invariants within the checked program. CHCs are accompanied with  powerful solvers that can determine their satisfiability. That is, they can find an appropriate interpretation for the uninterpreted predicates, so that the CHCs are satisfied. Given a verification problem consisting of a program and a safety property, the problem can be translated to a set of CHCs, which is satisfiable iff the program satisfies the property.

While being successful in proving safety properties, CHCs lack the power of proving  properties such as termination and general liveness  or branching temporal properties. For the former, a notion of well-foundedness is needed. For the latter, an existential quantifier is needed for choosing one of several branches.

\emph{Existential Horn Clauses (EHCs)} 
with their solver E-HSF~\cite{10.1007/978-3-642-39799-8_61},  
is an extension of CHCs, which include
existential quantification as well as the power of handling well-foundedness.

Our work is built on two pillars. One, is the seminal paper by Kesten and Pnueli~\cite{compositional}, which establishes a compositional approach to  deductive verification of First-Order $CTL^*$. 
The other is 
the EHC logic with its solver E-HSF~\cite{10.1007/978-3-642-39799-8_61} applicable to linear programs, which together provides a strong formal setting for verification problems that can be determined by EHC satisfiability.    

The compositional approach to $CTL^*$ in~\cite{compositional} provides proof rules for basic (assertional) state formulas of the form $Q c$, where $Q$ is a path quantifier and $c$ is an assertion. Formulas including temporal operators are reduced to the basic form by applying rules for successive eliminations of temporal operators. The reductions involve building program-like (testers) structures that are later composed with each other and with the verified program.

Our work adopts the idea of transforming a specification to basic state formulas while eliminating temporal operators. However, it takes a significantly different approach. We develop a translation system $\tool$, which given a verification problem consisting of a program $P$ and a specification $\phi$, builds a set of EHCs which is satisfiable iff  $P$ satisfies $\phi$. $\tool$ consists of 8 translation rules. The first replaces each basic subformula in $\phi$ with a boolean variable. The second replaces non-fair path quantifiers $Q$ with fair path quantifiers $Q_f$. The next three rules eliminate temporal path subformulas of the form $Next$ ($Xc$), $Globally$ ($Gx$), and $Until$ ($c_1 U c_2$),  using fresh boolean variables and  fresh uninterpreted predicates. Additional constraints require that the program satisfies the appropriate temporal property.

It should be noted, that our programs include fairness conditions that restrict the paths under consideration. This is because, the elimination of both $Globally$ and  $Until$ require such conditions.
Moreover, handling fairness as part of the program and not as part of a specification is much more efficient. 
Typically, the program is large and the specification is relatively small. However, 
specifications describing program fairness might be of size of the program.

The sixth and seventh translation rules in $\tool$ deviate the most from~\cite{compositional}. 
Rule $6$ handles the verification of $A_f c$, i.e., along every fair path, assertion $c$ holds. For that purpose, the translation imposes the requirement that if a path does not satisfy $c$, then it is not fair. Well-foundedness is needed for proving this fact. Rule $7$ handles $E_f c$, i.e., for every initial state there is a fair path starting at it. Here we need both well-foundedness and an existential quantification, both supplied by EHCs.  In our translation rules we also use negations and disjunctions that are not allowed by CHCs, but are expressible by EHCs.

We prove that our translation $\tool$ is sound and relative complete.

Next, we turn to exploiting our approach for synthesis. We are given a \emph{partial program},  in which some of the conditions (e.g. in while loops) and assignments are missing (\emph{holes}). We are also given a $CTL^*$ specification. Our goal is to synthesize a ``filling'' to the holes, so that the synthesized program satisfies the $CTL^*$ specification.

We obtain this by first filling each hole with an uninterpreted predicate. Additional constraints are added to guarantee that predicates associated with a condition or an assignment indeed behave accordingly.  
Finally, we build EHCs to describe the new program and send it to the EHC solver, which will return interpreted predicates to serve as hole fillings.
Here again we prove the soundness and relative completeness of our approach.

Finally, we present case studies to demonstrate the applicability of our algorithms for $CTL^*$ verification and synthesis.

\vspace{-0.2cm}
\subsubsection{Related work}
 
Several works~\cite{beyene2014ctlfo,DBLP:journals/corr/BeyenePR16}  exploit the approach presented in~\cite{compositional} to perform CTL+FO verification via EHC-solving. We significantly extend their result to perform full $CTL^*$ verification. We also show that EHCs can be further exploited for $CTL^*$ synthesis.

In~\cite{DBLP:conf/popl/BeyeneCPR14}, EHCs are used to solve reachability, safety and LTL games on infinite graphs. This tool can perform $CTL^*$ verification, since $CTL^*$ can be expressed through $\mu$-calculus~\cite{DAM199477}, $\mu$-calculus verification can be performed through parity-games~\cite{niwinski1996games}, which are a special case of LTL games. 
However, the translation of $CTL^*$ to  $\mu$-calculus is doubly exponential~\cite{DAM199477}. As a result, a direct model checking of $CTL^*$ according to our approach is more efficient than first translating and then solving the resulting LTL game.


In~\cite{Rothenberg2023}, condition synthesis is performed, given a program with holes in conditions and specifications as assertions in the code. The program behaviors are modeled ``backwards'' using CHCs.  The CHC solver infers conditions so that the synthesized  program satisfies the assertions.
In our work we synthesize both conditions and assignments with respect to full $CTL^*$. Program behaviors are modeled ``forward''. However, we use EHCs rather than the more standard modeling via CHCs to model the problem.

Many works use synthesis to fill holes in partial program. Sketch~\cite{10.1145/1168919.1168907,guo2022learning} synthesizes programs to satisfy input-output specifications. In template-based 
synthesis~\cite{srivastava2012template-based} the specification is in the form of assertions in the code. Overview on synthesis and on synthesis of reactive programs can be found in~\cite{PGL-010} and~\cite{inbook}.
Condition synthesis can also be used to \emph{repair} programs~\cite{7985681,inproceedings,10.1109/TSE.2016.2560811,Rothenberg2023}.

Termination has been proved in~\cite{10.1007/11817963_37}.
Several works are able to synthesize \emph{finite-state} full systems for \emph{propositional} CTL$^*$ specifications.
\cite{DBLP:journals/jacm/KupfermanVW00} translates propositional CTL$^*$ formulas to hesitant alternating automata.
%
\cite{DBLP:conf/cav/KhalimovB17} proposes  bounded synthesis for CTL$^*$. The paper bounds the number of states in the synthesized system to some $k$. The bound is increased if no such system exists. Clearly, such an algorithm is complete only if the goal of the synthesis is a finite-state system. Moreover, it cannot determine unrealizability.
\cite{Bloem_2017} shows how to apply CTL$^*$ synthesis via LTL synthesis. Both logics are propositional.
It exploits translation to a hesitant automaton, which is then sent to an SMT solver.

In contrast, our work handles First-Order CTL$^*$ specifications. It can synthesize infinite-state systems and can determine unrealizability. However, we are given as input a partial (sketch-like) program, from which we synthesize a full program.

Another type of Horn-like clauses, very similar to EHCs, called pfwCSP, has been introduced in~\cite{10.1007/978-3-030-81685-8_35}.
pfwCSP does not allow existential quantification in the heads of clauses but allows disjunctions and can require that an uninterpreted predicate behaves like a function. Instead of disjunctively well-founded clauses, pfwCSP can require well-foundedness.

The authors of \cite{10.1007/978-3-030-81685-8_35} present a fully automated CEGIS-based method for solving pfwCSP and implement it for the theory of quantifier-free linear arithmetic.
Furthermore, the authors describe how relational verification concerning different relational verification problems, including k-safety, co-termination, termination-sensitive non-interference, and generalized non-interference for infinite-state programs, can be performed using pfwCSP.

The authors conjecture that EHCs and pfwCSP are inter-reducible.

In future work, our rules can be redesigned to reduce $CTL^*$ verification to pfwCSP satisfiability by replacing the only occurrence of the existential quantifier in rule~\ref{rule7} with a Skolem function using a functional predicate and substituting disjunctive well-foundedness with well-foundedness.

Another use of pfwCSP is presented in~\cite{10.1145/3571265}. The authors define a first-order fixpoint logic with background theories called $\mu$CLP and reduce the validity checking of $\mu$CLP to pfwCSP solving. Moreover, they present a modular primal-dual method that improves the solving of the pfwCSP, generated for $\mu$CLP formulas.

The authors mention that several problems, such as LTL, CTL, and even full $\mu$-calculus model checking, can be encoded using $\mu$CLP for infinite-state systems in a sound and complete manner. Similar techniques were used in \cite{10.1007/978-3-030-81685-8_35} to perform $\mu$-calculus verification in a sound but incomplete way.

Another sound but incomplete procedure for $CTL^*$ verification of infinite-state systems has been presented in \cite{10.1007/978-3-319-21690-4_2}.

\section{Preliminaries}

\subsection{Forall-exists Horn-like clauses (EHC)}\label{sectionEHC}

Let $\TT$ be a first-order theory over signature $\Sigma$. 
Then,
\emph{Forall-exists Horn-like clauses}, also called \emph{Existential Horn Clauses} (EHCs), are first-order formulas of two types: implications and $dwf$-clauses.
\begin{enumerate}

\item

Implications are formulas with uninterpreted predicates of the form:

$$c(v_0)\wedge q_1(v_1)\wedge \dots \wedge q_n(v_n) \-> \exists w :b(w_0) \wedge p_1(w_1) \wedge\dots  p_m(w_m)$$
where $q_1,\ldots,q_n, p_1,\ldots,p_m$ are uninterpreted predicate symbols, not included in $\Sigma$; $v_0\dots v_n,w_0,\dots, w_m$ are tuples of variables, not necessarily disjoint; $c$ and $b$  are quantifier-free formulas from $\TT$ over variables $v_0$ and $w_0$, respectively, and $w\subseteq \bigcup_{i=0}^m w_i$. 
 
The left-hand side of the clause is called the \emph{body} of the clause and the right-hand side is called the \emph{head}.

\item
Disjunctively well-founded clauses or $dwf$-clauses are of the form: $$ dwf(q) $$

where $q(v,v')$ is an uninterpreted predicate of arity $2n$ for some $n\in \NN$.

\end{enumerate}

\subsubsection*{The semantics of EHCs}
Let $\II$ be an interpretation, defined over a domain $\DD$. $\II$  associates with each uninterpreted predicate of arity $k$ a $k$-ary relation over $\DD$. 
Given a formula $\phi$, $\II(\phi)$ is the formula obtained by replacing each occurrence of an uninterpreted predicate $p$ with its interpretation, $\II(p)$.
We say that clause $f(u) = body(u) \-> head(u)$ is satisfied under $\II$, if $\II(body(u)) \models_{\TT} \II(head(u))$.

 
To define the semantics of $dwf$-clauses we need the notion of a disjunctively well-founded relation. 
\begin{definition}
A relation $r (v, v')$ over a set $X$ is \emph{well-founded} if there is no  infinite sequence of elements in $X$, $x_1, x_2, \ldots$ such that for every $i$: $r(x_i, x_{i+1})$.
A relation $r(v,v')$ is \emph{disjunctively well-founded} if it is included in a finite union of well-founded relations. Formally, $r(v,v') \-> r_1(v,v')\vee \dots \vee r_l(v,v')$ for some well-founded relations $r_1,\dots, r_l$.
\end{definition}
%
An interpretation $\II$ 
 satisfies $dwf(q)$ for a predicate $q$ of arity $2n$, if $\II(q)$ is a disjunctively well-founded relation over the set $\DD^n$. 

Note that, our notion of satisfiability is semantic, referring only to the domain over which the interpretation $\II$ is defined, ignoring the syntax over which relations associated with uninterpreted predicates are expressed.


Theorem 1 from~\cite{1319598} can be reformulated as follows:
%
A relation is \emph{well-founded} if its transitive closure is disjunctively well-founded.
We  refer to this formulation when we use this theorem in our work.

\vspace{-0.3cm} 
\subsubsection*{Expressing disjunction and negation with EHC}
For our purposes, we need to express a disjunction in the heads of EHC clauses. A clause of the form
$body \-> head_1\vee head_2 $
 is equivalent to the following three forall-exists Horn-like clauses, using a new  boolean variable~$a$:
 \vspace{-0.22cm}
$$body \-> \exists a : head(a);\ \ \ head(0) \-> head_1; \ \ \ head(1) \-> head_2$$

Given a predicate $p(v)$, we define a predicate $q(v)$ such that $q(v) \equiv \neg p(v)$ using the following two clauses:
\vspace{-0.22cm} 
$$p(v)\wedge q(v) \-> \bot;\ \ \  \top \-> p(v) \vee q(v)$$
If both of them are satisfied, then $q(v) \equiv \neg p(v)$.

From now on, we use disjunctions and negations as part of the forall-exists Horn-like clauses, implying that we express them as described above.

\subsection{$CTL^*$ verification}
In this section we present the verification problems that we solve.


\noindent {\bf Program} Let us fix a first-order theory $\TT$.
We view a program as a transition system with fairness conditions. A program $P$ consists of a tuple of program variables $v$, and two formulas $next(v,v')$ and $init(v)$, describing the transition relation and the set of initial states, respectively. In addition, it includes a finite set $\JJ$ of assertions in $\TT$, describing the set of fairness conditions.

The states of $P$ are the valuations of $v$.  
A \emph{path} of $P$ is a sequence of states $s_1,s_2,\ldots$, such that $s_i,s_{i+1} \models_{\TT} next(v,v')$ for every $i$. An infinite path $s_1,s_2,\ldots$ is \emph{fair} if for every $\phi\in \JJ$ we have $s_i\models_\TT \phi$ for infinitely many $i$'s. 
If $\JJ = \emptyset$ then every infinite path is fair.
Note that the transition relation is not required to be total. Nevertheless, only infinite sequences of states are considered as program paths in the semantics of $CTL^*$ defined later.

\noindent {\bf Syntax and semantics of $CTL^*$}
We define the syntax of $CTL^*$ formulas in negation normal form (NNF), where negations are applied only to atomic assertions. Let $c$ range over assertions in $\TT$. State formulas $\phi$ and path formulas $\psi$ can be defined by the following grammar.
$$\phi:=c | \phi \wedge \phi | \phi \vee \phi | E\psi | A\psi | E_f \psi | A_f \psi $$
$$\psi:= \phi | G\psi | X\psi | \psi U \psi | \psi \wedge \psi | \psi \vee \psi$$
$CTL^*$ is the set of state formulas defined by this grammar.

We use $Q$ to denote a non-fair quantifier ($A$ or $E$), ranging over \emph{all} infinite program paths. $Q_f$ denotes a fair quantifier  ($A_f$ or $E_f$), ranging only over fair paths.


A $CTL^*$ state formula of the form $Q\psi$ or $Q_f\psi$ is called \textit{basic}, if $\psi$ does not contain path quantifiers. 

Let $P$ be a program, $s$ range over the states of $P$, and $\pi$ range over the paths of $P$. We use $\pi^i$ to denote the suffix of $\pi$, starting at the $i$'th state.

The semantics of $CTL^*$ is defined with respect to a program $P$, where boolean 
and temporal operators have their usual semantics: \ 
$\pi \models X\psi$ iff $\psi$ holds on $\pi^1$; \ $\pi \models G \psi$ iff $\psi$ holds on every suffix of $\pi$; \ $\pi \models \psi_1 U \psi_2$ iff $\psi_2$ eventually holds on a suffix $\pi^j$ of $\pi$ and $\psi_1$ holds in all preceding suffixes.
\begin{align*}
&P,s\models c \text{ iff } s\models_\TT c\\
&P,s\models E \psi \text{ iff there exists an infinite path $\pi$ starting at $s$ such that } P,\pi \models \psi\\
&P,s\models A \psi \text{ iff for every infinite path $\pi$ starting at $s$: } P,\pi \models \psi\\
&P,s\models E_f \psi \text{ iff there exists a fair path $\pi$ starting at $s$ such that } P,\pi \models \psi\\
&P,s\models A_f \psi \text{ iff for every fair path $\pi$ starting at $s$: } P,\pi \models \psi\\
&P,\pi \models \phi \text{ iff $s$ is the first state of $\pi$ and } P,s \models \phi\\
\end{align*}
%
%
\noindent {\bf The $CTL^*$ verification problem} We say that $P$ satisfies a $CTL^*$ formula $\phi$, denoted $P \models \phi$, if $\phi$ holds in every initial state of $P$. That is, for every state $s$ of $P$: $init(s) \-> P,s \models \phi$.

We view the verification problem as the tuple $(v, init(v), next(v,v'), J, \phi)$, where the first four elements are the parameters of the program and the last one is the $CTL^*$ state formula that the program must satisfy.

\section{Reduction of $CTL^*$ verification problem to EHC satisfiability}

Our goal is to generate a set of EHCs for a given verification problem, such that the set of EHCs is satisfiable iff the program meets the specification.

We denote the desired set of EHCs as $Clauses(v, init(v), next(v,v'), J, \phi)$ and define it recursively on the structure of $\phi$. We prove that this recursive definition is always sound (Theorem \ref{soundness}) and relatively complete (Theorem~\ref{completeness}).

Below, we present our \emph{translation system} $\tool$, which consists of eight rules, and explain the intuition behind them. 

We use the notation $\phi_1(\phi_2)$ to describe a formula $\phi_1$ that contains one or more occurrences of a state subformula $\phi_2$. $\phi_1(\gamma)$ then stands for the formula $\phi_1$ in which every occurrence of $\phi_2$ is replaced by the state formula~$\gamma$.

Denote $D = (v, init(v), next(v,v'), J)$. Recall that if $D \models \phi$ then for all $s$, $init(s) \-> \phi$.

\begin{enumerate}
\item\label{rule1}
Consider the $CTL^*$ state formula $\phi_2(\phi_1)$ with a basic state subformula $\phi_1$
 and let $aux$ be a  fresh uninterpreted predicate of arity $|v|$, then
$$Clauses(D, \phi_2(\phi_1)) =$$
$$=Clauses(v, aux(v), next(v,v'), J, \phi_1 ) \cup Clauses(D, \phi_2(aux)(v)).$$

{\bf Explanation of Rule 1:} The fresh predicate $aux$ represents an underapproximation of the set of states in which $\phi_1$ is satisfied. Since our formulas are in NNF, where negations are applied only to atomic assertions, we can conclude that $\phi_2(aux(v))$ is an underapproximation of $\phi_2(\phi_1)$. Thus, by $init(v) \-> \phi_2(aux)$ we also have  $init(v) \-> \phi_2(\phi_1)$.
\vspace{0.2cm}

\item\label{rule2}
Let $Q \phi$ be a $CTL^*$ basic state formula, where $Q$ is a non-fair path quantifier ($E$ or $A$). Let $Q_f$ be the fair version of $Q$ ($E_f$ or $A_f$, respectively), then
$$Clauses(D, Q\phi) = Clauses(v, init(v), next(v,v'), \emptyset, Q_f\phi).$$

{\bf Explanation of Rule 2:} Since the only path quantifier in $Q\phi$ is non-fair, it does not depend on the fairness conditions. Hence, we can eliminate those conditions from the program and replace $Q$ with its fair version.
\vspace{0.2cm}

\item\label{rule3}
Let $c$ be  an assertion in $\TT$. Suppose we have the $CTL^*$ basic state formula $Q_f \phi(Xc)$. Let $x_X$ be a new boolean variable, then
$$Clauses(D, A_f\phi(Xc)) = $$ $$=Clauses(v\cup \{ x_{X} \}, init(v), next(v,v')\wedge (x_{X} = c(v')), J,  A_f \phi(x_{X})).  $$
In the case of an existential path quantifier we also need a fresh uninterpreted predicate $aux$ of arity $|v|+1$:
$$Clauses(D, E_f\phi(Xc)) = \{ init(v) \-> \exists x_X : aux(v,x_X) \} \cup $$ 
$$ \cup \ Clauses(v\cup \{ x_{X} \}, aux(v, x_X), next(v,v')\wedge (x_{X} = c(v')), J,  E_f \phi(x_{X})). $$

{\bf Explanation of Rule 3:} The new variable $x_X$ represents the value of $c(v)$ in the next state of the path. This is achieved through the new constraint added to $next(v, v')$. Hence, we can replace $Xc$ by $x_X$ in the verified formula. 

For the existential quantifier $E_f$, we add a new auxiliary predicate $aux$ over $(v,x_X)$ that, if initially holds for a certain value of $x_X$, it is guaranteed that $E_f\phi(x_X)$ holds as well. The additional clause $init(v) \-> \exists x_X : aux(v,x_X)$, when interpreted together with the rest of the clauses,  guarantees that for every initial state there is a value of $x_X$, such that $aux(v, x_X)$ initially holds. The reason why we need this auxiliary predicate is that we do not know the value of $x_X$ in the beginning of the path. For example, if $ \phi(Xc) = FXc$, then a path that satisfies this formula may or may not have $ Xc$  in the first state. Therefore, we want to express that we can consider paths with any initial value of \( x_X \). The second set of clauses guarantees that if initially $aux(v,x_X)$ holds for some value $x_X$, then $F_f\phi(x_X)$ holds as well.

The same approach is applied in rules~\ref{rule4} and~\ref{rule5} with the existential quantifier.   

\vspace{0.2cm}

\item\label{rule4}
Let $c$ be  an assertion in $\TT$. Suppose we have the $CTL^*$ basic state formula $Q_f \phi(Gc)$.
Let $x_G$ be a new boolean variable, then 
$$Clauses(D, A_f \phi(Gc))= $$ $$ =  Clauses(v\cup \{ x_{G} \}, init(v), next(v,v')\wedge (x_{G} = (c(v)\wedge x_G')), J\cup \{ x_G\vee \neg c(v) \}, A_f \phi(x_{G})).  $$
In the case of an existential path quantifier, similarly to Rule~\ref{rule3}, we also need a fresh uninterpreted predicate $aux$ of arity $|v|+1$:
$$Clauses(D, E_f \phi(Gc))=  \{ init(v) \-> \exists x_G : aux(v,x_G) \} \cup $$
$$Clauses(v\cup \{ x_{G} \}, aux(v,x_G), next(v,v')\wedge (x_{G} = c(v)\wedge x_G'), J\cup \{ x_G\vee \neg c(v) \}, E_f \phi(x_{G})) $$ 

{\bf Explanation of Rule 4:}
If the new variable $x_G$ is true in some state, then the new condition added to the transition relation ensures that $c(v)$ and $x_G$ are true in every consecutive state.
Thus,  $x_G$ represents the value of $Gc$ along the path.
The new fairness condition ensures that if for every state of the path $c(v)$ is true, then $x_G$ must be true infinitely often along this path.

\vspace{0.2cm}

\item\label{rule5}
Let $c_1$ and $c_2$ be assertions in $\TT$. Suppose we have the $CTL^*$ basic state formula $Q_f \phi(c_1U c_2)$. Let $x_U$ be a new boolean variable, then 
$$Clauses(D, A_f \phi(c_1Uc_2))= Clauses(v\cup \{ x_{U} \}, init(v),  $$ $$next(v,v')\wedge (x_{U} = c_2(v) \vee (c_1(v)\wedge x_U')),J  \cup  \{ \neg x_U \vee c_2(v) \},  A_f \phi(x_{U})).  $$
In the case of an existential path quantifier, similarly to Rule~\ref{rule3}, we need a fresh uninterpreted predicate $aux$ of arity $|v|+1$:
$$Clauses(D, E_f \phi(c_1Uc_2)) =$$ 
$$=Clauses(v\cup \{ x_{U} \}, aux(v,x_U), next(v,v')\wedge (x_{U} = (c_2(v) \vee (c_1(v)\wedge x_U'))),$$ $$ J\cup \{ \neg x_U \vee c_2(v) \}, E_f \phi(x_{U})) \cup  
 \{ init(v) \-> \exists x_U : aux(v,x_U) \}. $$

{\bf Explanation of Rule 5:} The new variable $x_U$ represents the value of $c_1Uc_2$ along the path. The new condition added to the transition relation defines $U$ recursively. The new fairness condition ensures that if for every state of the path $x_U$ holds, then $c_2(v)$ must hold infinitely often along this path.

\vspace{0.2cm}

\item\label{rule6}
Let $c$ be an assertion in $\TT$. Suppose we have a basic $CTL^*$ state formula $A_f c$. Let $p$ be a fresh uninterpreted predicate of arity $|v|$ and let $r,t$ be  fresh uninterpreted predicates of arity $2|v|$. Let $v_0,\dots, v_{|J|}$ be $|J|+1$ copies of the set of program variables $v$. Suppose $J\neq \emptyset$, then,
$$Clauses(D, A_f c )\ = \  
 \{ init(v) \wedge \neg c(v)\-> p(v) ,\text{ } next(v,v')\wedge p(v) \-> p(v'),$$ $$  p(v_0)\wedge \bigwedge_{i=1}^{|J|} (t(v_{i-1},v_{i})\wedge J_i(v_i)) \-> r(v_0,v_{|J|}), \text{ }  dwf(r)  ,$$ $$ next(v,v') \-> t(v,v'), \ t(v,v')\wedge next(v',v'')\-> t(v,v'')\},$$
where $J_i$ is the $i$-th element of $J$.

If $J=\emptyset$, then,

$$Clauses(D, A_f c )= \ 
\ \{init(v) \wedge \neg c(v)\-> p(v) ,\text{ } next(v,v')\wedge p(v) \-> p(v'),$$
$$p(v)\wedge t(v,v')\-> r(v,v'),\text{ } dwf(r),$$
$$ next(v,v') \-> t(v,v'), \ t(v,v')\wedge next(v',v'')\-> t(v,v'')\}$$

\vspace{0.2cm}

{\bf Explanation of Rule 6:} The new predicate $p(v)$ represents states that are reachable from initial states that do not satisfy $c(v)$. The predicate $t$ represents the transitive closure of the transition relation $next$. The predicate $r(v,v')$ represents a pair of states from $p$ that can be connected by a path that satisfies every fairness constraint at least once.

In order to verify that the program satisfies $A_f c$ we need to prove that every initial state either satisfies $c(v)$ or there is no fair path starting in that state. 
The latter is guaranteed by showing that $r(v, v')$ is well-founded, which implies that not all fairness conditions hold infinitely often.

If $J = \emptyset$, then every infinite path is fair, since the program does not have any fairness constraints. Hence, we want to avoid infinite paths that start in initial states that do not satisfy $c$. In that case, we set $r(v,v')$ to represent pairs of states from $p$ that can be connected by a path. Since $r$ is disjunctively well-founded, there is no infinite path from a state that satisfies $p$.

\vspace{0.2cm}

\item\label{rule7}
Let $c$ be an assertion in $\TT$ and let $E_f c$ be a basic $CTL^*$ state formula. Let $q_1,\dots,q_{|J|}$ be fresh uninterpreted predicates of arity $|v|$ and  $r_1,\dots,r_{|J|}$ be fresh uninterpreted predicates of arity $2|v|$. Suppose $J\neq \emptyset$. Then,
$$Clauses(D, E_f c ) \ = \ 
 \{init(v)\-> c(v)\wedge q_1(v) \} \ \cup $$
$$\cup \ \{ q_i(v) \-> \exists v': next(v,v')\wedge ((J_i(v)\wedge q_{(i\% |J|) + 1}(v'))\vee(r_i(v,v')\wedge q_i(v'))), $$
$$  dwf(r_i) , \text{ }r_i(v,v') \wedge r_i(v',v'') \-> r_i(v,v'') \; | \; i\leq |J| \},$$

where $J_i$ is the $i$-th element of $J$. $i\% |J|$ is the remainder of $i$ modulo $|J|$.

If $J=\emptyset$, let $q$ be a fresh predicate of arity $|v|$, then
$$Clauses(D, E_f c ) \ = \ 
\{ init(v)\-> c(v)\wedge q(v) , \  q(v)\-> \exists v': next(v,v')\wedge q(v')\} $$

\vspace{0.2cm}

{\bf Explanation of Rule 7:} In order to verify that the program satisfies $E_f c$ we need to prove that every initial state $s_0$ satisfies $c(v)$ and there is a fair path starting at $s_0$.

The clauses $dwf(r_i) , \text{ }r_i(v,v') \wedge r_i(v',v'') \-> r_i(v,v'')$ ensure that $r_i$ is a disjunctively well-founded and transitive relation. Hence, it is well-founded. $r_i(v,v')$ means that $v'$ is closer to a state that satisfies $J_i$ than $v$. Since $r_i$ is well-founded it takes only finitely many steps to get to a state that satisfies $J_i$.

The new predicate $q_i$  represents the states from which it is possible to get to a state $s$ that satisfies $J_i$. In addition, there is a transition from $s$ to a state satisfying $q_j$ for $j = (i\%|J|)+1$, meaning that eventually, $J_j$ holds too. 


We use $(i \%|J|) +1$ instead of $i+1$ to emphasize that when a state that satisfies $J_{|J|}$ is found, we continue to look for a state that satisfies $J_1$, so every $J_i$ will be satisfied infinitely many times along the path.  

If $\JJ = \emptyset$, then every infinite path is a fair path. Hence, we just prove that every initial state satisfies $c(v)$ and is the start of an infinite path. 
Predicate $q$ represents the set of states, that lie on an infinite path.
\vspace{0.2cm}

\item\label{rule8}
Let $c$ be an assertion in $\TT$. Suppose we have the $CTL^*$ state formula $c$, then 

$$Clauses(v, init(v), next(v,v'), J, c )= \{ init(v) \-> c(v) \}.$$
\end{enumerate}

\begin{remark}\label{remarkNonUnique}
Applying $\tool$ to a verification problem may result in a non-unique set of clauses.
This may occur when the specification is a basic state formula that contains more than one path subformula.
For example, $E_f ( Xc_1\vee Gc_2)$. In such cases, these subformulas can be eliminated (using rules 3 and 4) in any order. 
As the soundness and relative completeness theorems prove, all sets produced by $\tool$ correctly represent the given verification problem.

Unless otherwise stated, $Clauses(v, init(v), next(v,v'), J, \phi )$ will refer to \emph{any} set of clauses that may be produced by $\tool$.

\end{remark}

For every verification problem,  $\tool$  terminates, since every rule calls procedure $Clauses$ for a problem with a smaller verification formula, except of rule~\ref{rule2}, which can be called only once for every non-fair path quantifier in the specification formula. 

\begin{theorem}[complexity] \label{complexity}
Let $(D,\phi)$ be a verification problem, with a program of size $m$, with $k$ fairness conditions, and a specification formula of size~$n$. Then $\tool$ produces a set $Clauses(D,\phi)$ consisting of $O(n(n+k))$ clauses with a total size of $O(n(n+k)(n+m))$.
\end{theorem}

\subsection{Soundness}

\begin{theorem}[\textit{soundness}] \label{soundness}

For every set $Clauses(v, init(v), next(v,v'), J, \phi )$ of clauses produced by $\tool$,  if the set is satisfiable, then the program $D = (v, init(v), next(v,v'), J )$ satisfies $\phi$.
 
\end{theorem}
    
\begin{proof}
We prove the theorem by induction on the size of the formula $\phi$. Recall that by definition, CTL$^*$ formulas are state formulas. As such, they are true in a program $D$ if all initial states of $D$ satisfy the formula. That is, $init(v) \rightarrow \phi$. We will use this fact in the proof.

\vspace{0.2cm}

\noindent {\bf Base}:
Rule \ref{rule8} forms the basis of the induction.
For an assertion $c$ in $\TT$, 
$Clauses(D, c )= \{ init(v) \-> c(v) \}$. If this set of clauses is satisfiable then, by definition of CTL$^*$ satisfiability, $D \models c$.
\vspace{0.2cm}

\noindent {\bf Induction step}:\\
{\bf Rule \ref{rule1}}:\\
Assume a set of clauses $Clauses(D, \phi_2(\phi_1))$ is satisfiable by an interpretation $\II$.
Then,  $\II$ satisfies $Clauses(v, aux(v), next(v,v'), J, \phi_1 )$ and by the induction hypothesis, $(v, \II(aux)(v), next(v,v'), J) \models \phi_1$. Thus, $I(aux)(v) \rightarrow \phi_1$. Also, $\II$ satisfies $Clauses(D, \phi_2(aux)(v))$ and by the induction hypothesis 
$D \models \phi_2(I(aux))$. Since $\phi_2$ is in Negation Normal Form (NNF), replacing all occurrences of $I(aux)$ in $\phi_2$ with $\phi_1$  results in a formula that is also true in $D$. We thus conclude that $D \models \phi_2(\phi_1)$, as required.

\vspace{0.2cm}
We present here the proof of Rule~\ref{rule4} (Universal and Existential cases). The proofs of Rules~\ref{rule3}, \ref{rule5} are presented in the Appendix \ref{append:sound}.  

\vspace{0.2cm}
\noindent {\bf Rule \ref{rule4}, Universal}: \\
Suppose $D= (v, init(v), next(v,v'), J)$ does not satisfy $A_f\phi(Gc)$. Then there is a fair path $\pi = s_1,s_2,\dots$ such that $D, s_1 \models init(v)$ and $D,\pi \not \models \phi(Gc)$.

Consider the path $\pi_G= s_1',s_2',\dots$ of the program $D_G =(v\cup \{ x_{G} \}, init(v), next(v,v')\wedge (x_{G} = c(v)\wedge x_G'), J\cup \{ x_G\vee \neg c(v) \})$, where $s_i' (v) = s_i (v)$ and $s_i'(x_G) = true$ iff for every $j\geq i:$ $s'_j \models c(v)$. 
Note that, once $x_G$ is set to true on $\pi_G$, $c(v)$ is set to true. Moreover, they both remain true forever from that point on. 

$\pi_G$ is a fair path of $D_G$ since the fairness conditions of $D$ do not depend on $x_G$. The additional fairness condition of $D_G$ holds as well, since if $c(v)$  holds from a certain point on along $\pi_G$ then $x_G$ holds from that point on. Thus, $x_G\vee \neg c(v)$ holds  infinitely often along $\pi_G$.

Also $D_G,\pi_G \not \models \phi(x_G)$, since $x_G$ is true on the same locations of path $\pi_G$ as the locations of $\pi$ where $Gc$ is true.

Hence, $D_G$ does not satisfy $A_f \phi(x_G)$. By the induction hypothesis, every set of clauses produced by $\tool$ for 
$Clauses(D_G, A_f \phi(x_{G})) $
is unsatisfiable. Since $Clauses(D, A_f\phi(Gc)) = Clauses(D_G, A_f \phi(x_{G}))$,
this implies that every set of clauses produced by $\tool$ for  $Clauses(D, A_f\phi(Gc))$ is unsatisfiable, as required.

\vspace{0.2cm}
\noindent {\bf Rule \ref{rule4}, Existential}: \\
Suppose a set $Clauses(v, init(v), next(v,v'), J, E_f \phi(Gc)) $ 
is satisfied by interpretation $\II$.
Denote $D_G= (v\cup \{ x_{G} \}, \II(aux)(v,x_G), next(v,v')\wedge (x_{G} = c(v)\wedge x_G'), J\cup \{ x_G\vee \neg c(v) \})$. Then, by Rule~\ref{rule4}, $Clauses(D_G,E_f \phi(x_{G})) $ is also satisfied by $\II$. By the induction hypothesis $D_G \models  E_f \phi(x_{G})$. 

Let $s$ be a state of $D =  (v, init(v), next(v,v'), J)$, such that $D, s\models init(v)$. Since the clause $ \{ init(v) \-> \exists x_G : aux(v,x_G) \}$ is satisfied by interpretation $\II$, then there is a state $s_G$ of $D_G$  such that $s_G(v)=s(v)$ and $D_G,s_G \models I(aux)(v, x_G)$. Moreover, there is a fair path $\pi_G = s_G,s'_2,s'_3,\dots$ of $D_G$ such that $D_G,\pi_G \models \phi(x_G)$. Consider the path of $D$ defined by $\pi = s,s_2,s_3\dots$, where $s_i(v)$ is the restriction of $s'_i$ to the variables from $v$. 
%
%
Then $D,\pi \models \phi(Gc)$. This is because $x_G$ is true on the same places of the path $\pi_G$ as the places of $\pi$ where $Gc$ is true.
Moreover, $\pi$ is a fair path.
Thus, $D$ satisfies $E_f \phi(Gc)$.

\vspace{0.2cm}
\noindent {\bf Rule \ref{rule6}}: \\
Denote $D = (v, init(v), next(v,v'), J)$.
Suppose by way of contradiction that a set $Clauses(D, A_f c )$
is satisfied by an interpretation $\II$
but  the program $D $ does not satisfy the formula $ A_f c $. Then, there is a fair path $\pi = s_0,s_1,s_2,\dots$, such that $D, s_0 \models init(v)$ and $D,\pi \not \models c$.

By the clauses $init(v) \wedge \neg c(v)\-> \II(p)(v) ,\text{ } next(v,v')\wedge \II(p)(v) \-> \II(p)(v')$ we know that for every $i:$ $s_i \models \II(p)(v)$.

Since $\pi$ is fair, each $J_j$ holds infinitely often along $\pi$. Moreover, there is an infinite increasing sequence of indexes $0<i_1<i_2<i_3<,\dots$, such that for every $k:$ $s_{i_k}\models J_{k\%|J|+1}$.

By the clause
$\II(p)(v_0)\wedge \bigwedge_{i=1}^{|J|} (\II(t)(v_{i-1},v_{i})\wedge J_i(v_i)) \-> \II(r)(v_0,v_{|J|})$ we get that $s_{i_{n_1|J|}}, s_{i_{n_2|J|}}\models \II(r)(v,v')$ for every pair of natural numbers $n_1,n_2$, provided $n_1< n_2$. Thus, $\II(r)$, which is disjunctive well-founded, contains an infinite transitive chain, but this is impossible due to Theorem 1 from \cite{1319598}. A contradiction.

\vspace{0.2cm}
\noindent {\bf Rule \ref{rule7}}: \\
Suppose a set of clauses $Clauses(D, E_f c )$ 
is satisfied by interpretation $\II$. 

In order to prove that $D$ satisfies $E_f c$ we need to show that every state that satisfies $init(v)$ also satisfies $c(v)$ and is a start of a fair path.
 
 Suppose for some state $s$ of $D$, $s\models init(v)$, then by clause $init(v)\-> c(v)\wedge \II(q_1)(v)$ we have $s\models c(v)$ and $s\models \II(q_1)(v)$.

 By clauses $\{  dwf(\II(r_i)), \text{ }\II(r_i)(v,v') \wedge \II(r_i)(v',v'') \-> \II(r_i)(v,v'') \; | \; i\leq |J| \}$ we know that for every $i\leq |J|$, relation $\II(r_i)$ is disjunctively well-founded and transitive, hence it is a well-founded relation by Theorem 1 from \cite{1319598}.

 For every $i\leq |J|$, by clause $ \II(q_i)(v) \-> \exists v': next(v,v')\wedge ((J_i(v')\wedge \II(q_{(i\%|J|)+ 1})(v'))\vee(\II(r_i)(v,v')\wedge \II(q_i)(v'))) $ and the fact that $\II(r_i)$ is a well-founded relation we can conclude that for every state $t$ of $D$ if $t\models \II(q_i)(v)$ then there is a path in $D$ from $t$ to some state $t'$ such that $t'\models J_i(v) $ and there is a transition from $t'$ to some state $t''$ such that $t''\models \II(q_{(i\%|J|)+ 1})(v)$. 

Thus, there is a fair path, that starts in $s$. \qed
 
\end{proof}

\subsection{Relative completeness}

Next we show that $\tool$ is relative complete. That is, if a verification problem holds, i.e., the program meets its specification, then every set of clauses, produced by $\tool$ for that problem, is satisfied by some interpretation of the uninterpreted predicates. However, this interpretation might not be expressible by a first-order formula in our logic.

\begin{theorem}[\textit{relative completeness}] \label{completeness}
    If $D = (v, init(v), next(v,v'), J )$ satisfies $\phi$ then every set produced by $\tool$ for $Clauses(D, \phi )$ is satisfiable by some interpretation of the uninterpreted predicates.
\end{theorem}
\begin{proof}
We prove the theorem by induction on the size of the verified formula.
\vspace{0.2cm}

\noindent {\bf Base}:\\
{\bf Rule \ref{rule8}}:
For assertion $c$ in $\TT$,  if $D \models c$ then 
$Clauses(D, c )= \{ init(v) \-> c(v) \}$ is satisfiable  by definition of CTL$^*$ semantics.
\vspace{0.2cm}

\noindent {\bf Induction step}:\\
{\bf Rule \ref{rule1}}:\\   
Assume $D \models \phi_2(\phi_1)$.
For the uninterpreted predicate $aux$, we choose as interpretation the predicate $\II(aux)$, which is true in exactly those states where $\phi_1$ is true.
Then, $(v, \II(aux)(v), next(v,v'), J)$ meets $\phi_1$ and $(v, init(v), next(v,v'), J)$ meets $\phi_2(\II(aux))$. By the induction hypothesis, the corresponding sets of clauses are satisfiable. Consequently, $Clauses(D, \phi_1(\phi_2))$ is satisfiable.

\vspace{0.2cm}
\noindent In the rest of the proof we omit some details of the induction steps. The full proof appears in the Appendix \ref{appen:complete}. 

\vspace{0.2cm}
\noindent {\bf Rule \ref{rule3}, Universal}: \\
If $D= (v, init(v), next(v,v'), J)$ meets $A_f\phi(Xc)$, then $D_X=(v\cup \{ x_{X} \}, init(v), next(v,v')\wedge (x_{X} = c(v')), J)$ meets $A_f \phi(x_{X})$. This is because for every fair path of $D_X$, the formula $x_{X}$ is true in exactly the same places where $Xc$ is true on the corresponding path of $D$. By the induction hypothesis, every set of clauses produced by $\tool$ for $Clauses(D_X, A_f \phi(x_{X}))$ is satisfiable. Consequently,  every set of clauses produced by $\tool$ for $Clauses(D, A_f\phi(Xc))$ is satisfiable, as required.

\vspace{0.2cm}

\vspace{0.2cm}
\noindent {\bf Rule \ref{rule6}}:\\
Suppose the program $D= (v, init(v), next(v,v'), J)$ satisfies $A_f c$. Recall that, $D \models A_f c$ implies that every initial state of $D$ either satisfies $c$ or has no fair path starting in it. We define an interpretation $\II$  as follows:
\vspace{-0.1cm}
\begin{itemize}
\item 
    $s\models \II(p)(v)$ iff there is no fair path starting at $s$.
\item 
    $s,s'\models \II(t)(v,v')$ iff there is a path from $s$ to $s'$.
\item 
    $s,s'\models \II(r)(v,v')$ iff $s\models \II(p)(v)$, $s'\models J_{|J|}(v')$ and there is a path from $s$ to $s'$ which visits a state  satisfying $J_1$, then a state satisfying $J_2$ and so on until it ends in $s'$ that satisfies $J_{|J|}$.
\end{itemize}
\vspace{-0.2cm}

\vspace{0.2cm}
\noindent {\bf Rule \ref{rule7}}:\\
Suppose the program $D= (v, init(v), next(v,v'), J)$ satisfies $E_f c$. This means that every initial state $s$ of $D$ satisfies $c$, and, in addition, $s$ is the start of a fair path. We say that a state $s$ of $D$ is \emph{fair} if there is a fair path starting at~$s$.
We define the interpretation $\II$ as follows. For every $i \in \{1, \ldots, |J| \}$:

\begin{itemize}
 \item
$s\models \II(q_i)(v)$ iff $s$ is a fair state.
\item
$s,s' \models \II(r_i)(v,v')$ iff $s'$ is a fair state and $s$ satisfies $J_i$ or $s'$ is the successor of $s$ on the shortest
path $\pi$ leading from $s$ to a fair state $s''$  satisfying $J_i$.\qed
\end{itemize}
\end{proof}

\section{Case study 1: $CTL^*$ verification of a robot system}

\begin{wrapfigure}{l}{0.45\textwidth}
    \input{figures/robot}
    \caption{Robots program R}
    \label{robots}
    \end{wrapfigure}

Suppose we are given a system with three robots placed on rational points of a two-dimensional plane. A user makes an ongoing interaction with the system by sending commands. The system processes the commands and moves the robots accordingly. 

We wish to verify two properties. The first property states that the three robots are never located in the same position at the same time. The second property states that for each pair of robots, the user can manipulate the robots in such a way that this pair meets infinitely often.

Program $R$, presented in Figure~\ref{robots}, describes the interaction between the user and the system. Program $R$ consists of a while-loop. At the beginning of each iteration, three variables $robot\_id$, $a$ and $b$ are assigned non-deterministically. These are the inputs from the user. Variable $robot\_id$ represents the identification of the robot that the user wishes to move. Variables $a$ and $b$ represent the moving instruction, given by the user to the selected robot.

Variables $x_i$ and $y_i$ represent the current position of robot $i$ on the plane.

If the user chooses robot $i$, then the system changes its position according to the linear transformation defined for robot $i$, based on $a$ and $b$.
For each robot, a different linear transformation is being used. For example, for robot 2 the transformation is $(a+b,-2a-2b)$.

Let us denote the set of variables of $R$ (including the variable $pc$) as $v$ and the transition relation of $R$ as $next(v,v')$.
The initial positions of the robots are given by the initial condition: 
$init(v) = (x_1 = y_1 = x_2 = y_2 = x_3 = 0 )\wedge (y_3 = 2)$.
The property we wish to verify is given by the $CTL^*$ formula $\phi$.
$$\phi = AG(safe(v)) \wedge EG(F meet_{1,2}(v)) \wedge EG (F meet_{2,3}(v)) \wedge EG (F meet_{1,3}(v)),$$
where the formula $safe(v)$ states that the three robots are not all in the same position on the plane. Formula $meet_{i,j}(v)$ states that robot $i$ and robot $j$ are in the same position\footnote{$safe(v) = (x_1\neq x_2) \vee (x_1\neq x_3) \vee (x_2\neq x_3) \vee (y_1\neq y_2) \vee (y_1\neq y_3) \vee (y_2\neq y_3)$ and $meet_{i,j}(v) =(x_i = x_j)\wedge (y_i = y_j)$.}.


The first conjunct of $\phi$ ensures that the three robots are never in the same position at the same time. The second conjunct ensures that for every pair of robots, there is a path along which they meet infinitely often.


Next, we demonstrate how to generate the set of EHCs for the verification problem of $R \models \phi$, using \tool.

We first need to use rule~\ref{rule1} four times, one for each basic state subformula. As the result, we get the following
$$Clauses(R,\phi) = Clauses(v,aux_1(v),next(v,v'), \emptyset, AG(safe(v)))\cup \dots$$  $$\dots \cup Clauses(v,aux_4(v),next(v,v'), \emptyset, EG(Fmeet_{2,3}(v))) \cup$$ $$   \cup Clauses(R, aux_1(v) \wedge aux_2(v) \wedge aux_3(v) \wedge aux_4(v) )$$
The first set of clauses requires that in every state where $aux_1(v)$ is true, subformula $AG(safe(v))$ is satisfied. Similar requirements are presented in the next three sets.
The last set guarantees that in the initial states of $R$ all four $aux_i$ predicates are satisfied. 

To deal with the first set we can use rule~\ref{rule2} to turn $A$ to $A_f$ and then rule~\ref{rule4} to eliminate the temporal operator. This results in
$$Clauses(v\cup \{ x_{G} \}, aux_1(v), next(v,v')\wedge (x_{G} = safe(v)\wedge x_G'), \{ x_G\vee \neg safe(v) \}, A_f x_{G}).  $$

Finally, we use rule~\ref{rule6} to get the explicit set of clauses.
In the Appendix \ref{app:robot} 
we explain  why the resulting set of clauses is satisfiable.
  
For the second, third and fourth sets, the generation of clauses is similar. We  present here only the fourth set of clauses.

Rewrite the formula to replace $F$ with $U$:
$EG(Fmeet_{2,3}(v)) = EG(True U meet_{2,3}(v))$.
Next, apply rule~\ref{rule2} to turn
$E$ to $E_f$. Then apply rule~\ref{rule5} to eliminate the temporal operator $U$. This results in the following set of clauses.
$$Clauses(v,aux_4(v),next(v,v'), \emptyset, EG(Fmeet_{2,3}(v))) =$$ $$= \{ aux_4(v) \-> \exists x_U : aux_5(v,x_U) \}  \cup Clauses(v\cup \{ x_{G} \},aux_5(v,x_U),$$  $$next(v,v') \wedge (x_{U} = (meet_{2,3}(v) \vee x_U')), \{ \neg x_U \vee meet_{2,3}(v) \}, EG(X_U)).$$
As the last two steps, we apply rule~\ref{rule4} to eliminate the temporal operator $G$ and rule~\ref{rule7} to get the explicit set of clauses.

\noindent More details are given in the Appendix \ref{app:robot}.

\section{$CTL^*$ synthesis from partial programs}


In this section, we investigate the problem of $CTL^*$ synthesis from partial programs, which include holes of two types -- condition holes and assignment holes.

As in the previous part, a program is presented as a transition system with fairness conditions. Additionally,  we have the variable $pc$, ranging over program locations. The set of program variables is now presented as $v=\{pc\} \cup v_r$,  where $v_r$ is the set of all program variables except $pc$.

\textit{A partial program} is a  pair $\langle P,H \rangle$, where $P$ is a program and $H$ is a finite set of holes in $P$. There are two types of holes.

\textit{A condition hole} is a triple of locations $\langle l, l_t, l_f \rangle$ -- the first element represents the location where we would like to insert a condition. If this condition is satisfied by a certain state at location $l$, then the synthesized program moves from location $l$ to location $l_t$, otherwise, it moves to location $l_f$.   

\textit{An assignment hole} is a pair of locations $\langle l, l'\rangle$. In this case, we want to synthesize an assignment in the form of a relation $\alpha(v_r,v_r')$ over pairs of states such that the synthesized program moves from location $l$ to location $l'$, assigning variables according to $\alpha$.

A \textit{resolving function} for a partial program $\langle P, H \rangle$ is a function $\Psi$ that maps every condition hole to a subset of $\DD^{|v_r|}$, and every assignment hole to a subset of $\DD^{|v_r|} \times \DD^{|v_r|}$\  \footnote{Recall that $\DD$ is the domain over which interpretations are defined.}, such that for every assignment hole $\langle l,l' \rangle $ the following holds:
$$\forall v_r \exists v_r': \Psi(\langle l,l' \rangle) (v_r,v_r')$$. 

The constraint above characterizes a non-deterministic multiple assignment.
Using constraints to characterize the assignment holes gives us flexibility in the synthesis.
For example, we could alternatively require that the assignment changes exactly one variable, while all others preserve their previous value.

For a partial program $\langle P,H \rangle$ and a resolving function $\Psi$ the \textit{synthesized program} is a program $P_\Psi= \langle v,init(v),next_\Psi(v,v'), \mathcal{J}\rangle$, where
$$next_{\Psi}(v,v') = \bigvee_{\langle l,l_t,l_f \rangle \in H} (pc=l\wedge pc'=l_t \wedge \Psi(\langle l,l_t,l_f \rangle)(v_r) \wedge v_r'=v_r) \vee$$ $$\bigvee_{\langle l,l_t,l_f\rangle\in H} (pc=l\wedge pc'=l_f \wedge \neg \Psi(\langle l,l_t,l_f \rangle)(v_r) \wedge v_r'=v_r)  \vee$$ $$\bigvee_{\langle l,l'\rangle \in H} (pc=l\wedge pc'=l' \wedge \Psi(\langle l,l' \rangle)(v_r,v_r')) 
\vee next(v,v').$$

\noindent So basically $P_\Psi$ is $P$ with the holes filled according to $\Psi$.

\textbf{The $CTL^*$ synthesis problem} is given a partial program $\langle P,H\rangle$ and a $CTL^*$ formula $\phi$. A solution to the $CTL^*$ synthesis problem is a resolving function $\Psi$, such that the synthesized program $P_\Psi$ satisfies $\phi$. If such a resolving function exists we say that the problem is \textit{realizable}. Otherwise, the problem is \textit{unrealizable}. 

\subsection{Reduction of $CTL^*$ synthesis to EHC satisfiability}

Our goal is to generate a set of EHCs for a given synthesis problem, such that the set of EHCs is satisfiable if and only if there is a solution to the synthesis problem, exactly as we did for the verification problem. 

Actually, the verification problem is a special case of the synthesis problem, so it is not  surprising that we use an extension of our approach to $CTL^*$ verification to perform $CTL^*$ synthesis. The main idea of our algorithm is to fill holes with uninterpreted predicates, apply our reduction for $CTL^*$ verification and add some additional clauses to the resulted set of EHCs, in order to set some constraints on uninterpreted predicates in holes.

Formally, for a given synthesis problem consisting of a partial program $\langle P,H\rangle$ and a
$CTL^*$ formula $\phi$, we define the set of uninterpreted predicates $U$ as follows:
$U := \{ u^c_{l}(v_r) | \langle l,l_t,l_f \rangle \in H \} \cup \{ u^a_l(v_r,v_r') | \langle l, l' \rangle \in H \}$.

Basically, $U$ contains a predicate over $v_r$ for every condition hole and a predicate over $v_r,v_r'$ for every assignment hole.

Intuitively, $u^c_{l}$  and $\neg u^c_{l}$ represent the sets of states from which the program moves from $l$ to $l_t$   and $l_f$, respectively. 

The predicate $u_l^a$ intuitively represents the set of pairs of states $s,s'$ (with locations $l$ and $l'$), such that the program moves from $s$ to $s'$. We need to require that for every state $s$ in $l$ there exists a state $s'$ in $l'$, such that the program moves from $s$ to $s'$. For that purpose, we introduce the following set of clauses
$\Delta_a(P,H) := \{ \top \-> \exists v_r': u^a_{l} (v_r,v_r') \; | \; \langle l,l'\rangle \in H \}$.

Let us define the program $P_U$, which is the program $P$ with holes filled with uninterpreted predicates. Formally $P_U= \langle v,init(v),next_U(v,v'), J 
\rangle$, where 
$$next_{U}(v,v') = next(v,v') \vee \bigvee_{\langle l,l_t,l_f \rangle \in H} (pc=l\wedge pc'=l_t \wedge u^c_{l}(v_r) \wedge v_r'=v_r) \vee$$  $$\vee \bigvee_{\langle l,l_t,l_f\rangle\in H} (pc=l\wedge pc'=l_f \wedge \neg u^c_{l}(v_r)  \wedge v_r'=v_r) \vee \bigvee_{\langle l,l'\rangle \in H} (pc=l\wedge pc'=l' \wedge u^a_{l}(v_r,v_r') )\footnote{Recall that we can express the negation of a predicate using EHCs. See Section~\ref{sectionEHC}}.$$

\noindent Finally, the desired set of clauses is 
$\Delta(P,H,\phi) := \Delta_a(P,H) \cup Clauses(P_U, \phi)$.
\vspace{0.1cm}
\begin{theorem}[soundness]\label{synthlemma}
    Suppose for a given $CTL^*$ synthesis problem $\langle P,H,\phi \rangle$, a set of clauses $\Delta(P,H,\phi)$ produced by $\tool$ is satisfiable by interpretation $\II$. Let $\Psi$ be a resolving function for $\langle P,H \rangle$ defined by:  

    For every condition hole $\langle l,l_t,l_f \rangle \in H$, \ 
    $\Psi(\langle l,l_t,l_f \rangle)(v_r) = \II(u^c_l)(v_r)$.

    For every assertion hole $\langle l,l' \rangle \in H$, \ 
    $\Psi(\langle l,l' \rangle)(v_r,v_r') = \II(u^a_{l})(v_r,v_r')$.

\noindent Then $P_\Psi$ satisfies $\phi$.
\end{theorem}

\begin{theorem}[Relative completeness] \label{syncompleteness}
    Given a realizable $CTL^*$ synthesis problem $\langle P,H,\phi \rangle$, then every set of clauses $\Delta(P,H,\phi)$ produced by $\tool$ is satisfiable.  
\end{theorem}

\noindent The proofs appear in the Appendix \ref{append:SyntProofs}.
\section{Case study 2: Synthesizing a bank-client interaction}

Suppose we would like to write a program that will determine the rules for an ongoing interaction between a client and a bank. The goal is to define a process that is satisfactory for both the client and the bank. Specifically, we would like to complete the given sketch on 
Figure~\ref{fig:bank}(a) according to a given specification written as a $CTL^*$ formula (see below).

The interaction  between the client and the bank is an ongoing iterative process, defined over the variables $req$ (for request), $bal$ (for balance), $exp$ (for expenses) and $pro$ (for profit). At each iteration the client may request to  apply either a withdrawal (in which case $req \geq 0$) or a deposit (where $req < 0$) to his account.
$bal$ represents the balance (amount of money) in the client's account. $exp$ represents the accumulative  withdrawals made by the client up to this iteration, and $pro$ holds the profit accumulated by the bank up to this iteration.

Initially, $bal = exp = pro = 0$.

According to the sketch in Figure~\ref{fig:bank}(a), after an unspecified check on line $l_c$, on line $l_3$ the client's balance is updated according to $req$. 
Then the computation is split to the case where request is a deposit (line $l_{a_1}$) or a withdrawal (line $l_5$). 
In both cases the associated assignments are unspecified, but
in the latter case, $exp$ is updated according to the current withdrawal.
Finally, on line $l_6$ the bank's profit is updated, decreasing it by \$$0.1$, regardless of the action that has been applied.
This is because the bank has to transfer a fee of \$$0.1$ to the central bank on each transaction it makes.

The CTL$^*$ specification formula $\phi$ is given as follows:
$$\phi = A(G(exp<1000) \vee FG(pro>50)) \wedge EF((exp \geq 100)\wedge pro< 7).$$

The first conjunct of $\phi$ requires that for every path either $exp$ is always less than \$$1000$ or from some moment on, $pro$ is greater than \$$50$ forever. It assures that if the amount of money withdrew by the client is sufficiently large, then $pro$ is also large.

The second conjunct assures that the client can withdraw at least \$$100$ from his account, while not paying the bank more than \$$7$. 

The entire $\phi$ ensures that the interaction between the client and the bank, performed according to the synthesized program, is acceptable for the bank (first conjunct) and for the client (second conjunct).

Note that the first conjunct also implies that the fees of \$$0.1$ transferred from the bank to the central bank (line $l_6$), are collected from the client. Otherwise, the bank would not be able to maintain  $pro>50$, as it might lose profit on client's actions.



The set of variables of $B$ is $v = \{ pc,req, bal, exp, pro \}$, the transition relation is $next(v,v')$ and the initial condition is $init(v) = (bal = exp = pro = 0)$. The set of holes is $H = \{l_c, l_{a_1}, l_{a_2} \}$.
We use $\tool$ to generate the set of clauses $\Delta(B,H,\phi)$ for the synthesis problem $\langle B, H, \phi \rangle$. 

The set of uninterpreted predicates is
$U = \{ u_{l_c}^c(v_r), u_{l_{a_1}}^a(v_r,v_r'), u_{l_{a_2}}^a(v_r,v_r') \}.$
The program $B_U= \langle v,init(v),next_U(v), \emptyset \rangle$, where 
$$next_{U}(v,v') =   (pc=l_c\wedge pc'=l_3 \wedge u^c_{l_c}(v_r) \wedge v_r'=v_r) \vee $$ $$ \vee   (pc=l_c\wedge pc'=l_2 \wedge \neg u^c_{l_c}(v_r)  \wedge v_r'=v_r) \vee $$  $$  \vee  (pc=l_{a_1}\wedge pc'=l_6 \wedge u^a_{l_{a_1}}(v_r,v_r') ) \vee (pc=l_{a_2}\wedge pc'=l_6 \wedge u^a_{l_{a_2}}(v_r,v_r') ) \vee next(v,v').$$
By definition,
$\Delta(B,H,\phi) = \Delta_a(B,H) \cup Clauses(B_U, \phi)$.
Since we have two assignment hole, 
$\Delta_a(P_T,H) = \{ \top \-> \exists v_r': u_{l_{a_1}}^a(v_r,v_r') , \top \-> \exists v_r': u_{l_{a_2}}^a(v_r,v_r') \}$.
The holes can be filled as presented in Figure~\ref{fig:bank}(b), resulting in program $B_S$.
     

    \begin{figure}%
    \centering
    \subfloat[\centering Partial program $B$]{{\input{figures/bankpartial} }}%
    \qquad
    \subfloat[\centering Synthesized program $B_S$]{{\input{figures/bankcomplete} }}%
    \caption{Synthesizing bank-client interaction}%
    \label{fig:bank}%
\end{figure}

The synthesized condition on line $l_c$ assures that the client has enough money to pay for his request and for the cost of transaction, \$$0.1$, so that the bank would not lose its profit.

On lines $l_{a_1}$ and $l_{a_2}$ of $B_S$ the bank collects \$$0.1$ from the client's account. Additionally, on line $l_{a_2}$ the bank collects $6\%$ of $req$, thus 
$pro$ is at least $6\%$ of $exp$ and it never decreases. The first conjunct of $\phi$ is then satisfied.

To see that the second conjunct holds, consider a scenario in which the client deposits \$$110$ and then withdraws \$$100$. After this transaction $exp = \$100$ and $pro = \$6$. Hence, the second conjunct of $\phi$ is satisfied.
Theorem~\ref{syncompleteness} now guarantees that the set $\Delta(B,H, \phi)$ is satisfiable.

\section{Conclusion}

In this work, we exploit the power of EHCs for the verification and synthesis of infinite-state programs with respect to $CTL^*$ specification.
The algorithms are both sound and relative complete. The relative completeness is based on semantic interpretation that might not be expressible in a specific syntax. However, this is the most general way to describe the desired interpretation.
For synthesis, for instance, when we conclude that the synthesis problem is unrealizable, then indeed no suitable solution exists, regardless of the syntax.

Our translation is independent of the solver or the underlying theory. If, in the future, a new EHC solver for any theory is designed, it can be combined with our translation system to perform $CTL^*$ verification and partial synthesis.

\newpage
\bibliographystyle{splncs04}
\bibliography{ppapers}

\clearpage
\appendix
\section{Appendix}

\subsection{Proof of complexity theorem}

\begin{theorem}[complexity] 
Let $(D,\phi)$ be a verification problem, with a program of size $m$, with $k$ fairness conditions, and a specification formula of size $n$. Then $\tool$ produces a set $Clauses(v,init(v),next(v,v'),J,\phi)$ consisting of $O(n(n+k))$ clauses with a total size of $O(n(n+k)(n+m))$.
\end{theorem}

\begin{proof}
At first, we prove that in the case when $\phi$ is a basic state formula, the number of clauses and the overall size are $O(n+k)$ and $O((n+k)(n+m))$ respectively. At the end of the generation process $\tool$ applies rule~\ref{rule6} or rule~\ref{rule7}. Rule~\ref{rule6} produces constant number of clauses, rule~\ref{rule7} produces $O(k)$ clauses and in both cases overall size is $O(n+mk)$. 

Before the last step $\tool$ applies rules~\ref{rule3}, \ref{rule4} and \ref{rule5}, replacing assertions inside the formula with boolean variables and extending the set of fairness conditions (by one element on each step) and transition relation (linearly from the size of the replaced assertion). Hence, by the last step of the algorithm, the sizes of the set of fairness constraints and the program are $O(n+k)$ and $O(n+m)$ respectively, hence the number of the clauses in the resulting set and the overall size are $O(n+k)$ and $O((n+k)(n+m))$. Also, during the implication of rules~\ref{rule3}, \ref{rule4} and \ref{rule5} $\tool$ may produce one clause per step of the size $m$, hence not more than $n$ clauses with overall size $n(n+m)$.

Rule~\ref{rule2} makes a recursive call of $Clauses$ for the problem of the same size, hence the bounds are preserved.

To deal with the case of non-basic specification formula $\phi$ we need to examine Rule~\ref{rule1}. It takes a state basic subformula of the specification formula of the size $l$ and makes two recursive calls. One for the problem with sizes of the parameters $m,n-l,k$ and one for the problem with sizes of the parameters $m,l,k$ and with the basic state specification formula. Eventually, after numerous implications of rule~\ref{rule1} $\tool$ reduces our problem to not more than $n$ subproblems with the same sizes of program and the set of fairness constraints and smaller specification state formulas. Hence, $\tool$ produces $O(n(n+k))$ clauses with the overall size $O(n(n+k)(n+m))$. \qed

\end{proof}

\subsection{Full proof of soundness theorem}\label{append:sound}

\begin{theorem}[\textit{soundness}] 

For every set of clauses produced by $\tool$ for $Clauses(v, init(v), next(v,v'), J, \phi )$,  if the set is satisfiable then the program $D = (v, init(v), next(v,v'), J )$ satisfies $\phi$.
 
\end{theorem}

\begin{proof}
We prove the theorem by induction on the size of the formula $\phi$. Recall that by definition, CTL$^*$ formulas are state formulas. As such, they are true in a program $D$ if all initial states of $D$ satisfy the formula. That is, $init(v) \rightarrow \phi$. We will use this fact in the proof.

\vspace{0.2cm}

\noindent {\bf Base}:
Rule \ref{rule8} forms the basis of the induction.
For an assertion $c$ in $\TT$, 
$Clauses(D, c )= \{ init(v) \-> c(v) \}$. If this set of clauses is satisfiable then, by definition of CTL$^*$ satisfiability, $D \models c$.
\vspace{0.2cm}

\noindent {\bf Induction step}:\\
{\bf Rule \ref{rule1}}:\\
Assume a set of clauses $Clauses(D, \phi_2(\phi_1))$ is satisfiable by an interpretation $\II$.
Then,  $\II$ satisfies $Clauses(v, aux(v), next(v,v'), J, \phi_1 )$ and by the induction hypothesis, $(v, \II(aux)(v), next(v,v'), J) \models \phi_1$. Thus, $I(aux)(v) \rightarrow \phi_1$. Also, $\II$ satisfies $Clauses(D, \phi_2(aux)(v))$ and by the induction hypothesis 
$D \models \phi_2(I(aux))$. Since $\phi_2$ is in Negation Normal Form (NNF), replacing all occurrences of $I(aux)$ in $\phi_2$ with $\phi_1$  results in a formula that is also true in $D$. We thus conclude that $D \models \phi_2(\phi_1)$, as required.

\vspace{0.2cm}
\noindent {\bf Rule~\ref{rule2}}:\\
Since the formula $Q\phi$ does not depend on the fairness conditions,
$(v, init(v), next(v,v'), J) \models Q\phi$ 
iff $(v, init(v), next(v,v'), \emptyset) \models Q_f\phi$.

Assume by way of contradiction that a set of clauses produced by $\tool$ for $Clauses(v, init(v), next(v,v'), J, Q\phi)$  is satisfiable but $(v, init(v), next(v,v'), J) \not\models Q\phi$. Then
$(v, init(v), next(v,v'), \emptyset) \not\models Q_f\phi$. By the induction hypothesis, since we consider $Q_f\phi$ to be smaller than $Q\phi$, every set of clauses produced by $\tool$ for $Clauses(v, init(v), next(v,v'), \emptyset, Q_f\phi)$ is unsatisfiable. Recall that Rule~\ref{rule2} defines $Clauses(v, init(v), next(v,v'), J, Q\phi)=Clauses(v, init(v), next(v,v'), \emptyset, Q_f\phi)$. Consequently, every set of clauses produced by $\tool$ for $Clauses(v, init(v), next(v,v'), J, Q\phi)$ is unsatisfiable as well. A contradiction.

\vspace{0.2cm}
\noindent We partition the proof for each of the rules \ref{rule3}, \ref{rule4} and \ref{rule5} into two cases --  universal and existential quantifiers.

\vspace{0.2cm}
\noindent {\bf Rule \ref{rule3}, Universal}: \\
We prove here that if $D= (v, init(v), next(v,v'), J)$ does not satisfy $A_f\phi(Xc)$ then for every set of clauses, produced by $\tool$ for $Clauses(v\cup \{ x_{X} \}, init(v), next(v,v')\wedge (x_{X} = c(v')), J,  A_f \phi(x_{X}))$, the set is unsatisfiable.

Suppose $D$ does not satisfy $A_f\phi(Xc)$. Then there is a fair path $\pi = s_1,s_2,\dots$ such that $D, s_1 \models init(v)$ and $D,\pi \not \models \phi(Xc)$.

Consider the path $\pi_X= s_1',s_2',\dots$ of the program $D_X = (v\cup \{ x_X \}, init(v), next(v,v')\wedge (x_{X} = c(v')), J)$, where $s_i' (v) = s_i (v)$ and $s_i'(x_X) = s_{i+1} (c(v))$. $\pi_X$ is a fair path of $D_X$, since the fairness conditions of $D$ and $D_X$ are identical and defined only over $v$ (and not over $x_X$). In addition, $D_X,\pi_X \not \models \phi(x_X)$, since $x_X$ and $Xc$ are true on exactly the same locations of the paths $\pi_X$ and $\pi$, respectively.

Hence, $D_X$ does not satisfy $A_f \phi(x_X)$. By the induction hypothesis, for every set of clauses produced by $\tool$ for $Clauses(v\cup \{ x_{X} \}, init(v), next(v,v')\wedge (x_{X} = c(v')), J,  A_f \phi(x_{X}))$, the set is unsatisfiable, as required.

\vspace{0.2cm}
\noindent {\bf Rule \ref{rule3}, Existential}: \\
Denote $D =  (v, init(v), next(v,v'), J)$. 
Suppose a set of clauses  $Clauses(D, E_f\phi(Xc))  = Clauses(v\cup \{ x_{X} \}, aux(v, x_X), next(v,v')\wedge (x_{X} = c(v')), J,  E_f \phi(x_{X})) \cup \{ init(v) \-> \exists x_X : aux(v,x_X) \}$ is satisfied by interpretation~$I$.

Denote $D_X= (v\cup \{ x_{X} \}, I(aux)(v, x_X), next(v,v')\wedge (x_{X} = c(v')), J)$. By the induction hypothesis, $D_X \models  E_f \phi(x_{X})$. 

Let $s$ be a state of $D$ such that $D, s\models init(v)$. Since the clause $ init(v) \-> \exists x_X : I(aux)(v,x_X) $ is true,  there is a state $s_X$ of $D_X$,  such that  $s_X(v)=s(v)$ and $D_X,s_X \models I(aux)(v, x_X)$. Thus, there is a fair path 
$\pi_X = s_X,s^X_2,s^X_3,\dots$ of $D_X$ such that $D_X,\pi_X \models \phi(x_X)$. Consider the path of $D$ defined as $\pi = s,s_2,s_3,\dots$, where each $s_i$ is the restriction of $s^X_i$ to the variables from $v$. Then $D,\pi \models \phi(Xc)$. This is because $x_X$ and $Xc$ are true on exactly the same locations of the paths $\pi_X$ and $\pi$, respectively.
We conclude that $D \models E_f \phi(Xc)$.

\vspace{0.2cm}
\noindent {\bf Rule \ref{rule4}, Universal}: \\
Suppose $D= (v, init(v), next(v,v'), J)$ does not satisfy $A_f\phi(Gc)$. Then there is a fair path $\pi = s_1,s_2,\dots$ such that $D, s_1 \models init(v)$ and $D,\pi \not \models \phi(Gc)$.

Consider the path $\pi_G= s_1',s_2',\dots$ of the program $D_G =(v\cup \{ x_{G} \}, init(v), next(v,v')\wedge (x_{G} = c(v)\wedge x_G'), J\cup \{ x_G\vee \neg c(v) \})$, where $s_i' (v) = s_i (v)$ and $s_i'(x_G) = true$ iff for every $j\geq i:$ $s'_j \models c(v)$. 
Note that, once $x_G$ is set to true on $\pi_G$, $c(v)$ is set to true. Moreover, they both remain true forever from that point on. 

$\pi_G$ is a fair path of $D_G$ since the fairness conditions of $D$ do not depend on $x_G$. The additional fairness condition of $D_G$ holds as well, since if $c(v)$  holds from a certain point on along $\pi_G$ then $x_G$ holds from that point on. Thus, $x_G\vee \neg c(v)$ holds  infinitely often along $\pi_G$.

Also $D_G,\pi_G \not \models \phi(x_G)$, since $x_G$ is true on the same locations of path $\pi_G$ as the locations of $\pi$ where $Gc$ is true.

Hence, $D_G$ does not satisfy $A_f \phi(x_G)$. By the induction hypothesis, every set of clauses produced by $\tool$ for 
$Clauses(D_G, A_f \phi(x_{X})) $
is unsatisfiable. Since $Closes(D, A_f\phi(Gc)) = Clauses(D_G, A_f \phi(x_{X}))$,
this implies that every set of clauses produced by $\tool$ for  $Closes(D, A_f\phi(Gc))$ is unsatisfiable, as required. 

\vspace{0.2cm}
\noindent {\bf Rule \ref{rule4}, Existential}: \\
Suppose a set $Clauses(v, init(v), next(v,v'), J, E_f \phi(Gc)) $ 
is satisfied by interpretation $\II$.
Denote $D_G= (v\cup \{ x_{G} \}, I(aux)(v,x_G), next(v,v')\wedge (x_{G} = c(v)\wedge x_G'), J\cup \{ x_G\vee \neg c(v) \})$. Then, by Rule~\ref{rule4}, $Clauses(D_G,E_f \phi(x_{G})) $ is also satisfied by $\II$. By the induction hypothesis $D_G \models  E_f \phi(x_{G})$. 

Let $s$ be a state of $D =  (v, init(v), next(v,v'), J)$, such that $D, s\models init(v)$. Since the clause $ \{ init(v) \-> \exists x_G : aux(v,x_G) \}$ is satisfied by interpretation $\II$, then there is a state $s_G$ of $D_G$  such that $s_G(v)=s(v)$ and $D_G,s_G \models \II(aux)(v, x_G)$. Moreover, there is a fair path $\pi_G = s_G,s'_2,s'_3,\dots$ of $D_G$ such that $D_G,\pi_G \models \phi(x_G)$. Consider the path of $D$ defined by $\pi = s,s_2,s_3\dots$, where $s_i(v)$ is a restriction of $s'_i$ to the variables from $v$. 
%
%
Then $D,\pi \models \phi(Gc)$. This is because $x_G$ is true on the same places of the path $\pi_G$ as the places of $\pi$ where $Gc$ is true.
Moreover, $\pi$ is a fair path.
Thus, $D$ satisfies $E_f \phi(Gc)$.

\vspace{0.2cm}
\noindent {\bf Rule \ref{rule5}, Universal}: \\
Suppose $D= (v, init(v), next(v,v'), J)$ does not satisfy $A_f\phi(c_1Uc_2)$, then there is a fair path $\pi = s_1,s_2,\dots$, such that $D, s_1 \models init(v)$ and $D,\pi \not \models \phi(c_1 U c_2)$.

Consider the path $\pi_U= s_1',s_2',\dots$ of the program $D_U = (v\cup \{ x_{U} \}, init(v), next(v,v')\wedge (x_{U} = c_2(v) \vee (c_1(v)\wedge x_U')), J\cup \{ \neg x_U \vee c_2(v) \})$, where $s_i' (v) = s_i (v)$ and $s_i'(x_U) = true$ iff $\pi^i$ satisfies $c_1 U c_2 $. $\pi_U$ is a fair path of $D_U$ since the fairness conditions from $D$ do not depend on the evaluation of $x_U$. Further, if $x_U$ is true from a certain point on along $\pi_U$, then $c_2$ must hold infinitely often. Thus, $\pi_U$ satisfies the additional fairness condition as well.
 
Also $D_U,\pi_U \not \models \phi(x_U)$, since $x_U$ is true on the same places of the path $\pi_U$ as the places of $\pi$ where $c_1 U c_2$ is true.

Hence $D_U$ does not satisfy $A_f \phi(x_U)$. By the induction hypothesis, every clause produced by $\tool$ for 
$ Clauses(v\cup \{ x_{U} \}, init(v), next(v,v')\wedge (x_{U} = c_2(v) \vee (c_1(v)\wedge x_U')), J\cup \{ \neg x_U \vee c_2(x) \},  A_f \phi(x_{U}))$ is unsatisfiable, as required.

\vspace{0.2cm}
\noindent {\bf Rule \ref{rule5}, Existential}: \\
Suppose a set $Clauses(v, init(v), next(v, v'), J, E_f\phi (c1 U c2))$
is satisfied by interpretation $I$.

Denote $D_U = (v\cup \{ x_{U} \}, init(v), next(v,v')\wedge (x_{U} = c_2(v) \vee (c_1(v)\wedge x_U')), J\cup \{ \neg x_U \vee c_2(x) \})$. By the induction hypothesis $D_U \models  E_f \phi(x_{U})$. 

Let $s$ be a state of $D =  (v, init(v), next(v,v'), J)$, such that $D, s\models init(v)$. Since the clause $ \{ init(v) \-> \exists x_X : aux(v,x_G) \}$ is satisfied by interpretation $I$, then there is a state  $s_U$ of $D_U$ such that $s_U(v)=s(v)$ and $D_U,s_U \models I(aux)(v, x_U)$. Moreover, there is a fair path $\pi_U = s_U,s'_2,s'_3,\dots$ of $D_U$ such that $D_U,\pi_U \models \phi(x_U)$. Consider the path of $D$ defined by $\pi = s,s_2,s_3,\dots$, where $s_i$ is the restriction of $s'_i$ to the variables from $v$. 


Then $D,\pi \models \phi(c_1 U c_2)$ since $x_U$ is true on the same places of path $\pi_U$ as the places of $\pi$ where $c_1 U c_2$ is true.
Hence, $D$ satisfies $E_f \phi(c_1 U c_2)$.

\vspace{0.2cm}
\noindent {\bf Rule \ref{rule6}}: \\
Denote $D = (v, init(v), next(v,v'), J)$.
Suppose by way of contradiction that a set $Clauses(D, A_f c )$
is satisfied by an interpretation $I$
but  the program $D $ does not satisfy the formula $ A_f c $. Then, there is a fair path $\pi = s_0,s_1,s_2,\dots$, such that $D, s_0 \models init(v)$ and $D,\pi \not \models c$.

By the clauses $init(v) \wedge \neg c(v)\-> I(p)(v) ,\text{ } next(v,v')\wedge I(p)(v) \-> I(p)(v')$ we know that for every $i:$ $s_i \models I(p)(v)$.

Since $\pi$ is fair, each $J_j$ holds infinitely often along $\pi$. Moreover, there is an infinite increasing sequence of indexes $0<i_1<i_2<i_3<,\dots$, such that for every $k:$ $s_{i_k}\models J_{k\%|J|+1}$.

By the clause
$I(p)(v_0)\wedge \bigwedge_{i=1}^{|J|} (I(t)(v_{i-1},v_{i})\wedge J_i(v_i)) \-> I(r)(v_0,v_{|J|})$ we get that $s_{i_{n_1|J|}}, s_{i_{n_2|J|}}\models r(v,v')$ for every pair of natural numbers $n_1,n_2$, provided $n_1< n_2$. Thus, $I(r)$, which is disjunctive well-founded, contains an infinite transitive chain, but that is impossible due to Theorem 1 from \cite{1319598}. A contradiction.

\vspace{0.2cm}
\noindent {\bf Rule \ref{rule7}}: \\
Suppose a set of clauses $Clauses(v, init(v), next(v,v'), J, E_f c )$ 
is satisfied by interpretation~$I$. 

Denote $D = (v, init(v), next(v,v'), J)$. In order to prove that $D$ satisfies $E_f c$ we need to show that every state that satisfies $init(v)$ also satisfies $c(v)$ and is a start of a fair path.
 
 Suppose for some state $s$ of $D$, $s\models init(v)$, then by clause $init(v)\-> c(v)\wedge I(q_1)(v)$ we have $s\models c(v)$ and $s\models I(q_1)(v)$.

 By clauses $\{  dwf(I(r_i)), \text{ }I(r_i)(v,v') \wedge I(r_i)(v',v'') \-> I(r_i)(v',v'') \; | \; i\leq |J| \}$ we know that for every $i\leq |J|$, relation $I(r_i)$ is disjunctively well-founded and transitive, hence it is a well-founded relation by Theorem 1 from \cite{1319598}.

 For every $i\leq |J|$, by clause $ I(q_i)(v) \-> \exists v': next(v,v')\wedge ((J_i(v')\wedge I(q_{(i\%|J|)+ 1})(v'))\vee(I(r_i)(v,v')\wedge I(q_i)(v'))) $ and the fact that $I(r_i)$ is a well-founded relation we can conclude that for every state $t$ of $D$ if $t\models I(q_i)(v)$ then there is a path in $D$ from $t$ to some state $t'$ such that $t'\models J_i(v) $ and there is a transition from $t'$ to some state $t''$ such that $t''\models I(q_{(i\%|J|)+ 1})(v)$. 

 Thus, there is a fair path, that starts in $s$. \qed
 
\end{proof}

\subsection{Full proof of relative completeness}\label{appen:complete}

Next we show that $\tool$ is relative complete. Relative completeness means that if a verification problem holds, that is, the program meets the specification, then every set of clauses, produced by $\tool$ for that problem, must be satisfied by some interpretation of the uninterpreted predicates. However, this interpretation is not necessarily expressible by a first-order formula in our logic.

\begin{theorem}[\textit{relative completeness}] 
    If $D = (v, init(v), next(v,v'), J )$ satisfies $\phi$ then every set produced by $\tool$ for $Clauses(D, \phi )$ is satisfiable by some interpretation of the uninterpreted predicate symbols.
\end{theorem}


\begin{proof}
We prove the theorem by induction on the size of the verified formula.
\noindent {\bf Base}:
Rule \ref{rule8} forms the basis of the induction.
For assertion $c$ in $\TT$,  if $D \models c$ then 
$Clauses(D, c )= \{ init(v) \-> c(v) \}$ is satisfiable  by definition of CTL$^*$ satisfiability.
\vspace{0.2cm}

\noindent {\bf Induction step}:\\
{\bf Rule \ref{rule1}}:\\   
Assume $D \models \phi_2(\phi_1)$.
For the uninterpreted predicate $aux$, we choose as interpretation the predicate $\II(aux)$, which is true in exactly those states where $\phi_1$ is true.
Then, $(v, \II(aux)(v), next(v,v'), J)$ meets $\phi_1$ and $(v, init(v), next(v,v'), J)$ meets $\phi_2(\II(aux))$. By the induction hypothesis, the corresponding sets of clauses are satisfiable. Consequently, $Clauses(D, \phi_1(\phi_2))$ is satisfiable.

\vspace{0.2cm}
\noindent{\bf Rule~\ref{rule2}}:\\
Since this verification problem does not depend on the fairness conditions, the following holds: $(v, init(v), next(v,v'), J) \models Q\phi$ \  iff \ $(v, init(v), next(v,v'), \emptyset) \models Q\phi$ iff $(v, init(v), next(v,v'), \emptyset) \models Q_f\phi$.
By the induction hypothesis (since we consider $Q_f\phi$ to be smaller than $Q\phi$), every set of clauses produced by $\tool$ for $Clauses(v, init(v), next(v,v'), \emptyset, Q_f\phi)$ is satisfiable. The same holds for $Clauses(v, init(v), next(v,v'), \JJ, Q\phi)$, since by Rule~\ref{rule2} is identical.

\vspace{0.2cm}
\noindent We partition the proof for each of the rules \ref{rule3}, \ref{rule4} and \ref{rule5} into two cases --  universal and existential quantifiers.

\vspace{0.2cm}
\noindent {\bf Rule \ref{rule3}, Universal}: \\
If $D= (v, init(v), next(v,v'), J)$ meets $A_f\phi(Xc)$, then $D_X=(v\cup \{ x_{X} \}, init(v), next(v,v')\wedge (x_{X} = c(v')), J)$ meets $A_f \phi(x_{X})$. This is because for every fair path of $D_X$, the formula $x_{X}$ is true in the exact the same places where $Xc$ is true on the corresponding path of $D$. By the induction hypothesis, every set of clauses produced by $\tool$ for $Clauses(D_X, A_f \phi(x_{X}))$ is satisfiable. Consequently,  every set of clauses produced by $\tool$ for $Clauses(D, A_f\phi(Xc))$ is satisfiable, as required.

\vspace{0.2cm}
\noindent {\bf Rule \ref{rule3}, Existential}: \\
Assume $D=(v , init(v), next(v,v'), J)$ satisfies $E_f \phi(Xc)$. We need to prove that every set of clauses produce by $\tool$ for 
$Clauses(D, E_f\phi(Xc)) =Clauses(v\cup \{ x_{X} \}, aux(v, x_X), next(v,v')\wedge (x_{X} = c(v')), J,  E_f \phi(x_{X})) \cup\{ init(v) \-> \exists x_X : aux(v,x_X) \} $ is satisfiable.

For $D_X=(v\cup \{ x_{X} \}, \II(aux)(v), next(v,v')\wedge (x_{X} = c(v')), J)$, we choose the predicate
     $\II(aux)$, which  is true in exactly those states of $D_X$ that satisfy $E_f \phi(x_{X})$. 
     
     Clearly, $D_X$ satisfies $E_f \phi(x_{X})$. By the induction hypothesis, every set of clauses produced by $\tool$ for $Clauses(D_X,  E_f \phi(x_{X}))$ is satisfiable. 

 For every fair path of $D_X$, the variable $x_{X}$ is true in exactly the same places where $Xc$ is true on the corresponding fair path of $D$.

Hence, $init(v) \-> \exists x_X  I(aux)(v,x_X)$, because if $s\models init(v)$ is true for some state $s$ of $D$, then there is a fair path $\pi$, starting at $s$ that satisfies $\phi(Xc)$. Hence, if $Xc$ is true on $\pi$, then $s\models I(aux)(v,1)$ and if $Xc$ is false on $\pi$ then $s\models I(aux)(v,0)$.

Thus, for every set of clauses produced by $\tool$ for   $Clauses(v\cup \{ x_{X} \}, aux(v, x_X), next(v,v')\wedge (x_{X} = c(v')), J,  E_f \phi(x_{X})) \cup \{ init(v) \-> \exists x_X : aux(v,x_X) \}$, we defined an interpretation that satisfy this set.

\vspace{0.2cm}
\noindent The two cases for rules~\ref{rule4} and~\ref{rule5}  can be proved similarly to Rule~\ref{rule3} and are omitted.

\vspace{0.2cm}
\noindent {\bf Rule \ref{rule6}}:\\
Suppose the program $D= (v, init(v), next(v,v'), J)$ satisfies $A_f c$. Recall that, $D \models A_f c$ implies that every initial state of $D$ either satisfies $c$ or has no fair path starting in it. We define an interpretation $\II$  as follows:
\begin{itemize}
\item 
    $s\models \II(p)(v)$ iff there is no fair path starting at $s$.
\item 
    $s,s'\models \II(t)(v,v')$ iff there is a path from $s$ to $s'$.
\item 
    $s,s'\models \II(r)(v,v')$ if $s\models \II(p)(v)$, $s'\models J_{|J|}(v')$ and there is a path from $s$ to $s'$ which visits a state  satisfying $J_1$, then a state satisfying $J_2$ and so on until it ends in $s'$ that satisfies $J_{|J|}$.
\end{itemize}
For every set of state $s,s',s'',s_0,\dots,s_{|J|}$  of $D$,  we check that all clauses defined by Rule~\ref{rule6} are satisfied by $\II$.
\begin{itemize}
\item
    $s\models init(v) \wedge \neg c(v)\-> \II(p)(v)$, because if an initial state does not satisfy $c$, then there is no fair path that starts in it. This is because $D$ satisfies $A_f c$.
\item
    $s,s'\models  next(v,v')\wedge \II(p)(v) \-> \II(p)(v')$, because if there is no fair path that starts in $s$ and there is a transition from $s$ to $s'$, then there is no fair path that starts in $s'$.
\item
    $s_0,\dots,s_{|J|} \models I(p)(v_0)\wedge \bigwedge_{i=1}^{|J|} (t(v_{i-1},v_{i})\wedge J_i(v_i)) \-> \II(r)(v_0,v_{|J|})$ by definition of $\II(r)$.
\item
    $s,s' \models next(v,v) \-> \II(t)(v,v')$ and $s,s',s'' \models next(v,v')
    \wedge \II(t)(v',v'')\-> \II(t)(v,v'') $, since $\II(t)$ is the transitive closure of $next$.  
\item    
    Relation $\II(r)$ is well-founded, since if there was an infinite chain of $\II(r)$, then this chain would represent a fair path on states satisfying $\II(p)$. But such states cannot occur on fair paths. Hence $dwf(\II(r))$. 
\end{itemize}
This concludes the proof of relative completeness for Rule~\ref{rule6}.

\vspace{0.2cm}
\noindent {\bf Rule \ref{rule7}}:\\
Suppose the program $D= (v, init(v), next(v,v'), J)$ satisfies $E_f c$. This means that every initial state $s$ of $D$ satisfies $c$, and, in addition, $s$ is the start of a fair path. We say that a state $s$ of $D$ is \emph{fair} if there is a fair path starting at~$s$.
We define the interpretation $\II$ as follows. For every $i \in \{1, \ldots, |J| \}$:

\begin{itemize}
 \item
$s\models \II(q_i)(v)$ iff $s$ is a fair state.
\item
$s,s' \models \II(r_i)(v,v')$ iff $s'$ is a fair state and $s$ satisfies $J_i$ or $s'$ is the successor of $s$ on the shortest
path $\pi$ leading from $s$ to a fair state $s''$  satisfying $J_i$.

\end{itemize}
For every two states $s, s'$ in $D$, we check that the clauses defined by Rule~\ref{rule7} are satisfied by $\II$.
\begin{itemize}
\item
    $s\models init(v)\-> c(v)\wedge \II(q_1)(v)$ since for every initial state $s$, $c(v)$ holds in $s$ and there is a fair path starting at $s$. That is, $s$ is fair.
\item
    $s\models  \II(q_i)(v) \-> \exists v': next(v,v')\wedge ((J_i(v)\wedge \II(q_{(i\% |J|) + 1})(v'))\vee(\II(r_i)(v,v')\wedge \II(q_i)(v')))$. This is because if $s\models \II(q_i)(v)$, then $s$ is a fair state. On a fair path every $J_j$ holds infinitely often. Thus, either $J_i(v)$ holds in $s$ or there is a shortest path $\pi$ from $s$ to a fair state where $J_i$ holds.
    
    In the first case we define $s'$ as the successor of $s$ on a fair path. In the second case we define $s'$ as the successor of $s$ on $\pi$.  Hence, we either have $s\models J_i(v)$ and $s'\models \II(q_{(i\% |J|) + 1})(v')$ (this is because fair states satisfy all $q_j$) 
    or $s'\models \II(q_i)(v)$ and $s,s'\models \II(r_i)(v,v')$.
  \item  
    $s\models \II(r_i)(v,v') \wedge \II(r_i)(v',v'') \-> \II(r_i)(v,v'')$, since the relation 
    is transitive by definition.
\item
    $\II(r_i)$ is a well-founded relation because if on each step we move to a state, which is closer along the shortest path to a state satisfying $(J_i \wedge q_{(i\% |J|)+1})$, we will eventually get to such a state. Hence $dwf(\II(r_i))$. 
 \end{itemize}  
This concludes the proof of relative completeness for Rule~\ref{rule7}. \qed
\end{proof}

\subsection{Case study 1: CTL$^*$ verification of user-robots system}
\label{app:robot}
Let us demonstrate how we generate the set of EHCs for the verification problem of $R \models \phi$ using \tool.

We need to use rule~\ref{rule1} four times, one for each basic state subformula. As the result, we get the following

$$Clauses(R,\phi) = Clauses(v,aux_1(v),next(v,v'), \emptyset, AG(safe(v)))\cup $$ $$\cup Clauses(v,aux_2(v),next(v,v'), \emptyset, EG(Fmeet_{1,2}(v))) \cup \ $$ $$ \cup Clauses(v,aux_3(v),next(v,v'), \emptyset, EG(Fmeet_{1,3}(v))) \cup $$ $$ \cup Clauses(v,aux_4(v),next(v,v'), \emptyset, EG(Fmeet_{2,3}(v))) \cup$$ $$   \cup Clauses(R, aux_1(v) \wedge aux_2(v) \wedge aux_3(v) \wedge aux_4(v) )$$

The first set requires that in every state where $aux_1(v)$ is true, subformula $AG(safe(v))$ is satisfied. Similar requirements are presented in the next three sets.
The last set guarantees that in the initial states of $R$ all four $aux_i$ predicates are satisfied. 

To deal with the first set we can use rule~\ref{rule2} to turn $A$ to $A_f$ and then rule~\ref{rule4} to eliminate the temporal operator. Finally, we get the following result.
$$Clauses(v,aux_1(v),next(v,v'), \emptyset, AG(safe(v))) = $$ $$ =  Clauses(v\cup \{ x_{G} \}, aux_1(v), next(v,v')\wedge (x_{G} = safe(v)\wedge x_G'),$$

$$ \{ x_G\vee \neg safe(v) \}, A_f x_{G}).  $$

Now we need to use rule~\ref{rule6} to get the explicit set of clauses.

$$Clauses(v\cup \{ x_{G} \}, aux_1(v), next(v,v')\wedge (x_{G} = safe(v)\wedge x_G'), \{ x_G\vee \neg safe(v) \}, A_f x_{G}) = $$
$$=  \{ aux_1(v) \wedge \neg x_G\-> p(v,x_G) ,\text{ } next(v,v')\wedge (x_{G} = safe(v)\wedge x_G') \wedge p(v,x_G) \-> p(v',x_G'),$$
$$p(v,x_G)\wedge t(v,x_G,v',x_G')\wedge (x_G'\vee \neg safe(v')) \-> r(v,x_G,v',x_G'),$$
$$ dwf(r), \ next(v,v') \wedge (x_{G} = safe(v)\wedge x_G') \-> t(v,x_G,v',x_G'), $$ $$ t(v,x_G,v',x_G')\wedge next(v',v'') \wedge (x_{G}' = safe(v')\wedge x_G'')\-> t(v,x_G,v'',x_G'')\}.$$

The resulting set of clauses guarantees that $aux_1(v)$ is evaluated in such a way that from every state in which $aux_1(v)$ is true, the user cannot direct three robots to the same point at the same time (since $safe(V)$ holds). 

Next, we explain informally why the set of clauses above is satisfiable.
Suppose from a state $s$ the user can apply a sequence of moves so that eventually all three robots arrive at position $(x;y)$. Hence $y=s(y_1)$, since the second coordination of robot 1 cannot be changed. For the second and third robots the sums $2x_2 + y_2$ and $x_3+y_3$ cannot be changed. Hence $$2x+s(y_1) = 2x+y = 2s(x_2)+s(y_2),$$ $$x = ( 2s(x_2)+s(y_2) - s(y_1))/2$$ and from the condition for the last robot, we need to have $ ( 2s(x_2)+s(y_2) - s(y_1))/2 + s(y_1) = s(x_3) + s(y_3)$. By these equations, we can conclude that $aux_1(v)$ can be evaluated as $$(( 2x_2+y_2 + y_1)/2)\neq x_3 + y_3$$
and $init(v)$ implies this condition.

For the second, third and fourth sets the generation of clauses in similar. We will look only at the fourth set.

Let us rewrite the formula to replace $F$ with $U$.
$$EG(Fmeet_{2,3}(v)) = EG(True U meet_{2,3}(v))$$
First, we need to apply rule~\ref{rule2} to turn
$E$ to $E_f$. Then we need to apply rule~\ref{rule5} to eliminate the temporal operator $U$.
$$Clauses(v,aux_4(v),next(v,v'), \emptyset, EG(True U meet_{2,3}(v))) = $$ $$ = \{ aux_4(v) \-> \exists x_U : aux_5(v,x_U) \} \cup$$ $$ \cup Clauses(v\cup \{ x_{G} \},aux_5(v,x_U),next(v,v') \wedge (x_{U} = meet_{2,3}(v) \vee x_U',$$ $$ \{ \neg x_U \vee meet_{2,3}(v) \}, EG(X_U))$$

As the last two steps, we need to apply rule~\ref{rule4} to eliminate the temporal operator $G$ and rule~\ref{rule7} to get the explicit set of clauses.

\subsection{Soundness and completeness of $CTL^*$ synthesis}\label{append:SyntProofs}
\begin{theorem}[soundness]\label{app:synthlemma}
    Suppose for a given $CTL^*$ synthesis problem $\langle P,H,\phi \rangle$, a set of clauses $\Delta(P,H,\phi)$ produced by $\tool$ is satisfiable by interpretation $\II$. Let $\Psi$ be a resolving function for $\langle P,H \rangle$ defined by:  

    For every condition hole $\langle l,l_t,l_f \rangle \in H$, \ 
    $\Psi(\langle l,l_t,l_f \rangle)(v_r) = \II(u^c_l)(v_r)$.

    For every assertion hole $\langle l,l' \rangle \in H$, \ 
    $\Psi(\langle l,l' \rangle)(v_r,v_r') = \II(u^a_{l})(v_r,v_r')$.

\noindent Then $P_\Psi$ satisfies $\phi$.
\end{theorem}

\begin{proof}
    First, we check that $\Psi$ is a resolving function. That is,  for every assignment hole $\langle l,l' \rangle \in H$, the formula 
    $\forall v_r \exists v_r': \Psi(\langle l,l' \rangle) (v_r,v_r') $ is true.
    But this follows from the fact that interpretation $I$ satisfies clauses $\Delta_a(P,H)$.

Next, we partition $I$ to $I_1$ and $I_2$, where $I_1$ is a restriction of $I$ to the set of predicates $U$ and $I_2$ is restricted to the predicates not in $U$.

We denote by $I_1(next_U)$ the formula $next_U$, in which every predicate $u^c_l$ or $u^a_l$ is interpreted according to $I_1$. Note that, by definition of $next_U,next_\Psi$ and $\Psi$, we get $I_1(next_U)(v,v')=next_\Psi(v,v')$.
Similarly, we denote by $I_1(P_U)$ the program $P_U$, in which the transition relation $next_U$ is replaced by its interpreted version $I_1(next_U)$. Note that, $I_1(P_U) = P_\Psi$. 

Applying $I_1$ to the predicates inside $Clauses(P_U,\phi)$ we get $$I_1(Clauses(P_U,\phi)) = Clauses(I_1(P_U),\phi) = Clauses (P_\Psi,\phi).$$ 
Since the set $Clauses(P_U,\phi)$ is satisfied by interpretation $I$ and $I=I_1\cup I_2 $, then $Clauses (P_\Psi,\phi) = I_1(Clauses(P_U,\phi))$ is satisfied by $I_2$. Thus, by Theorem~\ref{soundness} the program $P_\Psi$ satisfies $\phi$, as required. 
\qed

\end{proof}

\begin{theorem}[Relative completeness] \label{app:syncompleteness}
    Given a realizable $CTL^*$ synthesis problem $\langle P,H,\phi \rangle$, then every set of clauses $\Delta(P,H,\phi)$ produced by $\tool$ is satisfiable.  
\end{theorem}

\begin{proof}
    Suppose $\langle P,H,\phi \rangle$ is realizable by resolving function $\Psi$.

    \noindent As before, we partition the satisfying interpretation $I$ that we would like to define into two interpretations, $I_1$ and $I_2$, where the first   interprets the predicates from $U$ and the second interprets the rest.
    We define interpretation $I_1$ on $U$ as follows:

    For every condition hole $\langle l,l_t,l_f \rangle \in H$, \  
    $I_1(u^c_l)(v_r):= \Psi(\langle l,l_t,l_f \rangle)(v_r)$.

    For every assertion hole $\langle l,l' \rangle \in H$, \ 
    $I_1(u^a_{l})(v_r,v_r'):= \Psi(\langle l,l' \rangle)(v_r,v_r')$.

    We adopt the notations of $I_1(next_U), I_1(P_U)$ and $I_1(Clauses(P_U,\phi))$ from the proof of Lemma \ref{synthlemma}.  Note that, $I_1(next_U)(v,v')=next_\Psi(v,v')$, hence $P_\Psi = I_1(P_U)$. Thus, by Theorem~\ref{completeness} the set $Clauses(P_\Psi, \phi) = Clauses(I_1(P_U),\phi) = I_1(Clauses(P_U,\phi))$ is satisfiable by some interpretation $I_2$, since $P_\Psi$ satisfies $\phi$. Hence, $I:= I_1\cup I_2$ is a satisfying interpretation for $Clauses(P_U,\phi)$.

    Since the formula
    $\forall v_r \exists v_r': \Psi(\langle l,l' \rangle) (v_r,v_r') $
    is true, then $\Delta_a(P,H)$ is satisfied by $I_1$, and therefore also by $I$. Thus, $\Delta(P,H,\phi)$ is satisfied by $I$.\qed
    
\end{proof}


\end{document}